\documentclass[12pt, draftclsnofoot, onecolumn]{IEEEtran}
\usepackage{amsmath,amssymb,epsfig,amsthm,psfrag,epsf, multirow,wrapfig, graphicx}
\usepackage[english]{babel}
\usepackage{epstopdf}
\usepackage{cite}
\newtheorem{theorem}{Theorem}
\newtheorem{lemma}{Lemma}

\newtheorem{remark}{Remark}
\usepackage{cleveref}
\usepackage{ dsfont }
\usepackage{bm}
\usepackage{caption}
\usepackage{subcaption}
\usepackage{color}

\title{Multi-Cell Multi-User Massive FD-MIMO: Downlink Precoding and Throughput Analysis}
\author{Rubayet Shafin and Lingjia Liu
\thanks{R. Shafin and L. Liu are with the Bradley Department of Electrical and Computer Engineering, Virginia Tech, Blacksubrg VA, 24061, USA. The corresponding author is L. Liu (ljliu@ieee.org). This work has been submitted to the IEEE for possible publication. Copyright may be transferred without notice, after which this version may no longer be accessible.}
}
\begin{document}
\maketitle

%\thispagestyle{firstpage}
%\newpage

\begin{abstract}
In this paper, downlink (DL) precoding and power allocation strategies are identified for a time-division-duplex (TDD) multi-cell multi-user massive Full-Dimension MIMO (FD-MIMO) network.
Utilizing channel reciprocity, DL channel state information (CSI) feedback is eliminated and the DL multi-user MIMO precoding is linked to the uplink (UL) direction of arrival (DoA) estimation through estimation of signal parameters via rotational invariance technique (ESPRIT).
Assuming non-orthogonal/non-ideal spreading sequences of the UL pilots, the performance of the UL DoA estimation is analytically characterized and the characterized DoA estimation error is incorporated into the corresponding DL precoding and power allocation strategy.
Simulation results verify the accuracy of our analytical characterization of the DoA estimation and demonstrate that the introduced multi-user MIMO precoding and power allocation strategy outperforms existing zero-forcing based massive MIMO strategies.
\end{abstract}


\section{Introduction}

Massive-MIMO or large-scale MIMO, has generated significant interest both in academia
\cite{Noncooperative_Cellular} and industry \cite{Full_Dimension_MIMO}.
%with the promise of fulfilling future throughput demands by offering increased spectral-efficiency obtained via aggressive spatial multiplexing.
Because of the promise of fulfilling future throughput demand via aggressive spatial multiplexing, massive MIMO is considered as one of the key enabling technologies  for next generation wireless networks.
Due to form factor limitation at the base station (BS), three dimensional (3D) massive-MIMO/Full Dimension MIMO (FD-MIMO) systems have been introduced in 3GPP to deploy active antenna elements in a two dimensional (2D) antenna array enabling the exploitation of the degrees of freedom in both azimuth and elevation domains. 
%Accordingly, 3D massive-MIMO is also called full-dimension MIMO (FD-MIMO) in 3GPP LTE-Advanced systems.
{ Due to the availability of the huge spectrum in the millimeter wave (mmWave) band, mmWave communication is considered as another enabling technology for future cellular networks: 5G and beyond. However, due to its significantly higher path loss compared to the microwave channel, it is extremely challenging to establish an effective communication for outdoor channels using mmWave bands. This challenge can be tackled using beamforming techniques where the base station serves multiple users with narrower beams. This can be possible if a large number of antennas are deployed at the base station in order to realize the narrow beams. As a result, massive MIMO is a natural counterpart for the mmWave cellular network.}
{\color{red} }

Since the benefits of massive MIMO or massive FD-MIMO are limited by the accuracy of the downlink (DL) channel state information (CSI) available at the base station, it is critical for the BS to obtain corresponding DL CSI information.
In general, the BS can obtain the CSI knowledge through the following: 1) DL CSI feedback where the the CSI information is fed back from mobile stations (MSs), and 2) DL/UL channel reciprocity where BS estimates the uplink (UL) CSI and infers DL CSI information through channel reciprocity.
Note that DL CSI feedback is heavily used in frequency-division-duplex (FDD) systems where only a few bits of the corresponding DL CSI information are fed back to the BS \cite{Downlink_MIMO} to achieve a good tradeoff between DL MIMO performance and UL feedback overhead/reliability.
To utilize the DL/UL channel reciprocity, the critical point becomes estimating the UL channel at the BS.
Based on UL pilots/reference signals sent from MSs, there are generally two methods to estimate the UL channel.
First is to estimate the corresponding channel transfer function (e.g., UL channel matrix).
Alternatively, UL direction of arrival (DoA) can be estimated at the BS using ESPRIT algorithm~\cite{A_subspace_rotation_approach}.
Even though the DoA only provides partial information on the UL channel, it is shown in \cite{Rubayet_Journal, Angle_and_Delay_Estimation_for_3D} that it can be directly linked to DL MIMO precoding in TDD systems.
%Therefore, this provides us a new framework to conduct DL MIMO precoding and resource allocation for TDD massive FD-MIMO systems free of DL CSI feedback and UL pilot overhead.
It is important to note that the DoA based MIMO precoding strategy has also been introduced to FDD systems demonstrating significant performance benefits in reality \cite{Eliminating_channel_feedback}.

%{\color{blue}\sout{Despite promising better performance, non-linear precoding, such as dirty paper coding (DPC) or vector perturbation are not feasible for MIMO systems because of high implementation complexity.}}
{Despite promising better performance, non-linear precoding schemes, such as dirty paper coding (DPC) or vector perturbation, are not practical for MIMO systems due to its high implementation complexity. }
{In recent years, simple linear processing techniques have been shown to offer significant performance gains for multi-user massive MIMO scenarios where the base stations employ a large number of antennas\cite{Noncooperative_Cellular}.
}
{Hence, most of the prior works in massive MIMO literature have focused on maximum ratio transmission (MRT) and zero forcing (ZF)-based methods for DL MIMO precoding \cite{Massive_MIMO_for_Maximal_Spectral, Joint_Power_Allocation_and_User}.} However, for mmWave  mssive FD-MIMO systems, it is possible to design low-complexity precoders with better performance than the conventional ZF/MRT-based precoders.

{ In this paper, we will characterize the optimal/near-optimal DL MIMO precoding and power allocation strategies for a TDD multi-cell multi-user mmWave massive FD-MIMO network.
}
ESPRIT-based UL DoA estimation scheme will be introduced and performance of the DoA estimation will be analytically characterized assuming non-orthogonal spreading sequences used for UL pilots/reference signals.
DL multi-user MIMO precoding and power allocation strategies will be identified based on the UL DoA estimation and their corresponding error performance.
Performance evaluation will be conducted to illustrate benefits of the introduced MIMO precoding and power allocation strategy over our previous scheme~\cite{Rubayet_Journal} as well as popular zero forcing (ZF)-based precoding~\cite{Massive_MIMO_for_Maximal_Spectral,Joint_Power_Allocation_and_User}.
The contribution of this paper can be summarized as the following:
\begin{itemize}
    \item { First, we present a unitary
    	 ESPRIT-based uplink DoA estimation method for multi-cell multi-user mmWave massive FD-MIMO OFDM network.
    }
Unlike majority of existing work, our scheme considers a more realistic scenario where non-orthogonal spreading sequences are used as UL pilots for both intra-cell and inter-cell users. As a result, due to the non-zero correlation coefficients among users' spreading sequences, the UL DoA/channel estimation is subject to intra-cell interference, inter-cell interference, and the so called pilot contamination.
    \item Second, we analytically characterize the mean square error (MSE) of {unitary ESPRIT-based} UL DoA estimation for the corresponding multi-cell multi-user FD-MIMO network. Our analytical results show how different perturbation components, namely noise elements, and intra-cell and inter-cell interferences, affect the UL DoA/channel estimation performance. The MSE has been related to key physical parameters such as number of antennas, BS array geometry, complex path gains, and correlation coefficients between users' spreading sequences.
    \item Third, we derive the sum-rate maximizing DL precoding and power allocation strategy for our FD-MIMO system. Furthermore, we perform a large antenna array regime analysis for DL precoding and identify the optimum power allocation under both perfect and imperfect DoA estimation scenarios.
    \item Finally, we validate our algorithms and analytical results through extensive simulation. The evaluation results demonstrate that our simulated MSE for different antenna numbers and antenna array geometries match well with those of analytical expressions for both elevation and azimuth estimation in large SNR regimes.
        Moreover, we also show that the introduced sum-rate maximization precoding strategy outperforms both eigenbeamforming and ZF-based precoding over all SNR regimes.
\end{itemize}

The rest of the paper is organized as follows. Section II describes the system model and the channel model for the underlying multi-cell multi-user massive FD-MIMO network.
Section III presents the ESPRIT-based UL DoA estimation method and the performance characterization for DoA estimation.
The achievable sum-rate analysis under both perfect and imperfect DoA estimation as well as the optimal MIMO precoding and power allocation strategies are contained in Section IV.
Simulation results are presented in Section V and Section VI concludes the paper.

\section{System and Channel Model}\label{System_Model}
\begin{figure}[h!]
\centering
\includegraphics[width=0.8\linewidth]{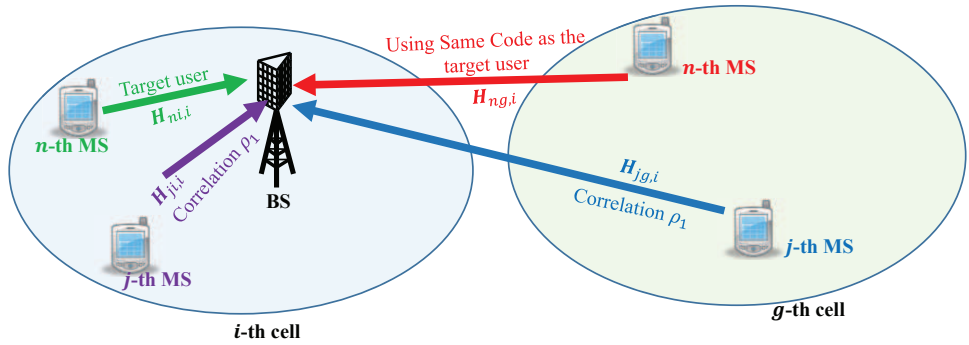}
%\captionsetup{margin= {20pt},justification=centerlast,skip=-5pt,font=normalsize}
%\vspace{0.1cm}
\caption[width=0.5\linewidth]{{Network Model.}}
%\captionsetup{justification=centering}
\label{system}
\end{figure}
We consider a multi-cell multi-user MIMO-OFDM system consisting of $G$ BSs as depicted in figure \ref{system}.
Each BS with $N_r$ number of antennas supports $J$ number of mobile stations (MSs)--each having $N_t$ number of transmit antennas.
After appending cyclic prefix (CP), the resulting time domain transmit signal at each MS is first passed through a parallel-to-serial converter followed by a digital-to-analog (DAC) converter, resulting in the baseband OFDM signal.
The baseband signal is then up-converted and sent through a  frequency selective fading channel, which is assumed to remain time-invariant during an OFDM symbol duration.
{It is to be noted here that we assumed same number of antennas at all UEs because of better clarity of exposition. However, we want to emphasize here that the algorithm and analysis presented in this work are not restricted by this assumption. All the results presented in this work can be straightforwardly extended to the scenario where users in the cell have different number of antennas.
}

In the UL, each MS sends $N_t$ spreading sequences of length $Q$ as plots/reference signals: one on each transmit antenna.
Accordingly, the $N_r \times Q$ frequency-domain received signal for the $k$-th subcarrier at the $i$-th BS can be written as
%\small
\begin{align}\label{Z_i}
\mathbf{Z}_i(k)=\sum\limits_{g=0}^{G-1}\sum\limits_{j=0}^{J-1}\sqrt{\Lambda_{jg,i}} \mathbf{H}_{jg,i}(k)\mathbf{X}_{jg}(k)+\mathbf{W}_i(k),
\end{align}
%\normalsize
where $\mathbf{H}_{jg,i}(k)$ is the $N_r \times N_t$ {channel matrix} for the channel between the $i$-th BS and the $j$-th MS in the $g$-th cell at the $k$-th subcarrier, and $\Lambda_{jg,i}$ is the corresponding large scale fading coefficient which is independent of subcarrier frequency; $\mathbf{X}_{jg}(k)$ is the $N_t \times Q$ frequency domain transmit signal from the $j$-th MS in the $g$-th cell for the $k$-th subcarrier, and $\mathbf{W}_i(k)$ is the corresponding $N_r \times Q$ noise matrix.
Note that each row vector of $\mathbf{X}_{jg}(k)$ is a length-$Q$ spreading sequence.
The channel transfer function, $\mathbf{H}_{jg,i}(k)$, can be written as
%\small
\begin{align}\label{H_jgi}
\mathbf{H}_{jg,i}(k)&=\sum_{\ell=0}^{L_{jg,i}-1}\mathbf{C}_{jg,i}(\ell)e^{\frac{-j2\pi k \ell}{N_c}},
\end{align}
%\normalsize
where $\mathbf{C}_{jg,i}(\ell)$ is the $N_r \times N_t$ channel impulse response (CIR) for the $\ell$-th tap of the channel between $i$-th BS and the $j$-th MS in the $g$-th cell.
{$N_c$ denotes total number of subcarriers.
}
Here, we assume that the channel, which can be represented by an equivalent discrete-time linear channel impulse response, has a finite number ($L_{jg,i}$) of non-zero taps.

Using the geometric channel model for mmWave frequencies, the impulse response for the $\ell$-th tap of the channel between  $i$-th BS and the $j$-th MS in the $g$-th cell can be represented by\cite{Rubayet_Journal, Angle_and_Delay_Estimation_for_3D, akdeniz2014millimeter}
%\small
\begin{align}\label{C_jgi}
\mathbf{C}_{jg,i}(\ell)= \sum\limits_{p=0}^{P_{jg,i,\ell}-1}\alpha_{jg,i}(\ell,p) \mathbf{e}_{r,jg,i}(\ell,p)\mathbf{e}_{t,jg,i}^H(\ell,p),
\end{align}
%\normalsize
where $\alpha_{jg,i}(\ell,p)$, $\mathbf{e}_{r,jg,i}(\ell,p)$, and $\mathbf{e}_{t,jg,i}(\ell,p)$ are, respectively, the channel gain, $N_{r} \times 1$ receive antenna array response, and $N_{t} \times 1$ transmit antenna array response for the $p$-th sub-path within the $\ell$-th tap of the channel between the $i$-th BS and the $j$-th MS in $g$-th cell;
$P_{jg,i,\ell}$ is the total number of sub-paths within the $\ell$-th tap of the channel;
and $(\cdot)^H$ denotes Hermitian transpose operation.
%It is obvious that the transmit and receive antenna array responses depend on directions of departure (DoD) of the transmit signal and directions of arrival (DoA) of the received signal, respectively.
In the FD-MIMO network of interests, a 1D uniform linear array (ULA) is assumed at each MS.
The corresponding transmit antenna array response can be described using the Vandermonde structure:
$\mathbf{e}_{t,jg,i}(\ell,p)= \begin{bmatrix}
1 & e^{j\omega_{jg,i,\ell,p}} & \ldots & e^{j(N_t-1)\omega_{jg,i,\ell,p}}
\end{bmatrix}^T $, where $\omega_{jg,i,\ell,p}=(2\pi\Delta_t/\lambda)\cos\Omega_{jg,i,\ell,p}$,  $\Delta_t$ is the spacing between the adjacent transmit antenna elements, $\Omega_{jg,i,\ell,p}$ is the transmit angle (DoD) for $p$-th sub-path within $\ell$-th tap of the channel between $i$-th base station and the $j$-th user in $g$-th cell, and $\lambda$ is the carrier wavelength.
On the other hand, for FD-MIMO networks the antenna array at the BS is a 2D planar array placed in the X-Z plane, with $M_{1}$ and $M_{2}$ antenna elements in vertical and horizontal directions, respectively.
Accordingly, the number of total receive antenna elements at the base station is $N_{r} = M_{1} \times M_{2}$.
%Since the antenna elements at the base station are placed in a 2D plane, for each resolvable path, there will be an azimuth DoA and an elevation DoA.
Therefore, the receive antenna array response for the $p$-th sub-path within $\ell$-th tap can be expressed as $\mathbf{e}_{r,jg,i}(\ell,p)=\mathbf{a}(v_{jg,i,\ell,p})\otimes\mathbf{a}(u_{jg,i,\ell,p})$, where $\otimes$ represents the Kronecker product, and  $\mathbf{a}(u_{jg,i,\ell,p})= \begin{bmatrix}
1 & e^{ju_{jg,i,\ell,p}} & \ldots & e^{j(M_{1}-1)u_{jg,i,\ell,p}}
\end{bmatrix}^T$ and $\mathbf{a}(v_{jg,i,\ell,p})= \begin{bmatrix}
1 & e^{jv_{jg,i,\ell,p}} & \ldots & e^{j(M_{2}-1)v_{jg,i,\ell,p}}
\end{bmatrix}^T$  can be treated as the receive steering vectors in the elevation and azimuth domains, respectively.
Here, $u_{jg,i,\ell,p}=\frac{2\pi \Delta_r}{\lambda}\cos\theta_{jg,i,\ell,p}$ and $v_{jg,i,\ell,p}=\frac{2\pi \Delta_r}{\lambda}\sin\theta_{jg,i,\ell,p} \cos\phi_{jg,i,\ell,p}$ are the two receive spatial frequencies at the BS, $\Delta_r$ is the spacing between adjacent antenna elements in the receive antenna array, and $\theta_{jg,i,\ell,p}$ and $\phi_{jg,i,\ell,p}$ are the elevation and azimuth DoAs for the $p$-th sub-path within $\ell$-th tap for the channel between the $i$-th BS and the $j$-th MS in $g$-th cell, respectively. 
{In this paper, we are not considering user mobility or scheduler impact on the system performance, and we will address these important issues in our future work.}

\section{Uplink Channel Characterization Through DoA Estimation}
%In this section, we will present the DoA estimation procedure for the multi-cell multi-user system and characterize estimation performance for 3D massive MIMO scenario.
In this section, we will present the UL DoA estimation procedure for the multi-cell multi-user massive FD-MIMO network and characterize assuming non-orthogonal/non-ideal spreading sequences and characterize the corresponding estimation performance. 
{We choose ESPRIT-based method for DoA estimation over other high resolution DoA estimation methods, such as MUSIC, since ESPRIT offers better resolvability and unbiased estimates with lower variance. Most importantly, ESPRIT provides significant computational advantages in terms of faster processing speed, lower storage requirement and indifference to knowledge of precise array geometry.
}
\subsection{UL DoA Estimation through Unitary ESPRIT}\label{UL_DOA_Estimation}
%In this sub-section, we present an ESPRIT-based uplink DoA estimation method for multi-cell MU-MIMO system model described in Section II.
Let the $n$-th MS at the $i$-th cell be the target user which tries to communicate to the $i$-th BS.
In massive FD-MIMO networks, the number of scheduled users may be quite large, and hence due to limited availability of orthogonal spreading codes, it may not be possible to assign orthogonal sequences to all scheduled users.
With this in mind, in this work, we assume a more realistic scenario that only the spreading sequences used by the same MS
are orthogonal while spreading sequences for different MSs within a cell are non-orthogonal.
Furthermore, we assume that the same pool of spreading sequences are reused across all cells as UL pilots complying with the 3GPP LTE/LTE-Advanced standards~\cite{Holma}.

Let the correlation among the spreading sequences from different MSs be denoted as $\rho_1$.
Now, for estimating the UL channel of the $n$-th MS in $i$-th cell, at the $i$-th BS, the $N_r \times Q$ received signal at the $k$-th subcarrier, $\mathbf{Z}_i(k)$, is first correlated with the spreading sequences of $n$-th MS.
Hence, after correlating the received signal with the target user's sequence, from \eqref{Z_i}, we have
%\small
\begin{align}\label{Z_i_k_2}
&\mathbf{Z}_i(k)\mathbf{X}_{ni}^H(k)=\sum\limits_{g=0}^{G-1}\sum\limits_{j=0}^{J-1}\sqrt{\Lambda_{jg,i}} \mathbf{H}_{jg,i}(k)\mathbf{X}_{jg}(k)\mathbf{X}_{ni}^H(k)+\mathbf{W}^{'}_i(k),
\end{align}
%\normalsize
where $\mathbf{W}^{'}_i(k)=\mathbf{W}_i(k)\mathbf{X}_{ni}^H(k)$ is the equivalent noise element.
%In this section, for the ease of exposition, we focus on the scenario where DoAs of the interference signals-- either inter-cell or intra-cell interference-- are not aligned with the DoAs of the target user's channel. Because of the highly directional nature of the mmWave communication, this situation ideally depicts the propagation environment where the number of scattering elements is relatively small. Impact of angular overlapping on the DoA estimation in section III.
Now, we can re-write \eqref{Z_i_k_2} as
%\small
\begin{align}\notag
\mathbf{Z}_i(k)\mathbf{X}_{ni}^H(k)&=\sqrt{\Lambda_{ni,i}}\mathbf{H}_{ni,i}(k)+\sum\limits_{\substack{g=0\\g \neq i}}^{G-1}\sqrt{\Lambda_{ng,i}}\mathbf{H}_{ng,i}(k)+\sum\limits_{\substack{j=0\\j \neq n}}^{J-1}\sqrt{\Lambda_{ji,i}}\mathbf{H}_{ji,i}(k)\rho_1\mathbf{1}_{N_t}\\\label{abcd}
&+\sum\limits_{\substack{g=0\\g \neq i}}^{G-1}\sum\limits_{\substack{j=0\\j \neq n}}^{J-1}\sqrt{\Lambda_{jg,i}}\mathbf{H}_{jg,i}(k)\rho_1\mathbf{1}_{N_t}+\mathbf{W}_i^{'}(k),
\end{align}
%\normalsize
where $\mathbf{1}_{N_t}$ denotes an $N_t \times N_t$ matrix with each element being unity.
In \eqref{abcd}, the first term, $\sqrt{\Lambda_{ni,i}}\mathbf{H}_{ni,i}(k)$,  represents the target user's UL channel; first summation term represents the inter-cell interference caused by users in other cells whose spreading sequences are exactly the same as that of target user (pilot contamination); second summation term represents the intra-cell interference; and third summation term represents the inter-cell interference caused by users in other cells whose spreading sequences are different (non-orthogonal) than that of target user.
%{\color{red}Correlation coefficient, $\rho_1$, in Equation-5 is a system design parameter in order to consider the tradeoff between the training sequence length and corresponding system performance. It is to be noted that  results in the special case where all the users in a cell are assigned orthogonal codes. 
%}
{
In realistic wireless networks such as LTE/LTE-Advanced networks, there exists a nonzero correlation between different pilot sequences. For example, Zadoff-Chu sequence is used to make sure different prime length Zadoff-Chu sequence has constant cross-correlation \cite{popovic1992generalized}. To reflect this practical constraint as well as provide system design insights, a correlation coefficient, $\rho_1$,  in \eqref{abcd} is introduced as a system design parameter to consider the tradeoff between the training sequence length and corresponding system performance. It is to be noted that  $\rho_1=0$ results in the special case where all the users in a cell are assigned orthogonal codes. Also important to note that the value of $\rho_1$ depends on the length of scrambling sequences the system designer chooses, which again depends on the coherence time of the channel. This provides us a way to investigate the impact of channel coherence on the network performance: Smaller coherence time will lead to a shorter scrambling sequence resulting in higher values for $\rho_1$.

}

Because of the large path loss, which is manifested by the large-scale fading coefficients, overall gains of the inter-cell interference channels are relatively small compared to that of the target user's channel. Furthermore, presence of the small correlation coefficients, $\rho_1$, intra-cell interference terms can also be considered relatively smaller compared to the term for target user's channel.
Hence, during the ESPRIT-based parameter estimation phase, we can treat the interference and noise elements together, and %write
%As a result, we can treat the interference channels as noise and lump these terms with noise element, $\mathbf{W}^{'}_i(k)$. Hence,
\eqref{abcd} can be written as
%\small
\begin{align}\label{Z_i_2}
\hat{\mathbf{H}}_{ni,i}(k)=\mathbf{Z}_i(k)\mathbf{X}_{ni}^H(k)=\sqrt{\Lambda_{ni,i}}\mathbf{H}_{ni,i}(k)+\mathbf{W}_i^{''}(k)
\end{align}
%\normalsize
where $\mathbf{W}_i^{''}(k)$ is the  equivalent noise-plus-interference matrix.
Now, using \eqref{H_jgi} and \eqref{C_jgi}, we can write $\mathbf{H}_{ni,i}(k)$ as
%\small
\begin{align}\label{H_nik}
&\mathbf{H}_{ni,i}(k)=\sum_{\ell=0}^{L_{ni,i}-1}\sum\limits_{p=0}^{P_{ni,i,\ell}-1}\alpha_{ni,i}(\ell,p) \mathbf{e}_{r,ni,i}(\ell,p)\mathbf{e}_{t,ni,i,k}^H(\ell,p)
\end{align}
%\normalsize
where $\mathbf{e}_{t,ni,i,k}(\ell,p)=\mathbf{e}_{t,ni,i}(\ell,p)e^{\frac{-j2\pi k \ell}{N_c}}$.
In order to jointly estimate the elevation and azimuth angles of the uplink channel between the $i$-th base station and the $n$-th user in $i$-th cell, we can now apply a low-complexity DoA estimation algorithm based on unitary ESPRIT.

High frequency channels, especially millimeter-wave channels, usually have fewer number of scattering clusters \cite{Broadband}.
In this work, we focus on the simple case where each scattering cluster contributes a single propagation path.
This is a reasonable assumption for the analysis of FD-MIMO systems~\cite{Rubayet_Journal, Spatially_Sparse, Channel_Estimation_and_Hybrid_Precoding}.
Hence for the clarity of exposition, and notational convenience, we can drop the subpath index, $p$, from $\alpha_{jg,i}(\ell,p) $, $\mathbf{e}_{r,jg,i,k}(\ell,p)$, and $\mathbf{e}_{t,jg,i}^H(\ell,p)$.
{
However, our results also hold  for multiple subpaths scenario due to the fact that ESPRIT can be used to distinguish subpaths as long as the spatial resolvability of the array is higher than the angular spread between two subpaths \cite{wang2012low}. 
It is to be noted here that in this work, instead of Standard ESPRIT, we utilize Unitary ESPRIT \cite{haardt1995unitary}  for DoA estimation, which provides superior estimation performance for the case where the sub-paths within the same clusters are highly correlated. Moreover, because of Forward-Backward Averaging (FBA), Unitary ESPRIT can still estimate the corresponding DoAs of two sub-paths which are completely correlated or coherent. It is also noteworthy here that it is unlikely to have more than two completely coherent sub-paths in the mmWave propagation channel based on 3GPP mmWave channel model \cite{3GPP_Study_on_channel_model} and the seminal work in \cite{samimi20163}. Therefore, our introduced algorithm is applicable for most general mmWave channels. However, for the very special case where more than two sub-paths are completely coherent, and all such sub-path DoAs are required to be estimated, the spatial smoothing technique can be applied in conjunction with FBA to de-correlate the corresponding signals \cite{pillai1989forward}. However, this is out of the scope of current manuscript and we will consider this special case in our future work.
}
Now, \eqref{H_nik} can be written as
%\small
\begin{align}
\mathbf{H}_{ni,i}(k)&=\mathbf{A}_{ni,i}\mathbf{D}_{ni,i}\mathbf{B}_{ni,i}^H(k)
\end{align}
%\normalsize
where $\mathbf{A}_{ni,i}=\begin{bmatrix}
\mathbf{e}_{r,ni,i}(0) &  \ldots & \mathbf{e}_{r,ni,i}(L_{ni,i}-1)
\end{bmatrix}$, $\mathbf{D}_{ni,i}=\text{diag}\begin{bmatrix}
\alpha_{ni,i}(0) & \ldots & \alpha_{ni,i}(L_{ni,i}-1)
\end{bmatrix}$,  and $\mathbf{B}_{ni,i}(k)=\begin{bmatrix}
\mathbf{e}_{t,ni,i,k}(0) & \ldots & \mathbf{e}_{t,ni,i,k}(L_{ni,i}-1)
\end{bmatrix}$.
Hence, from \eqref{Z_i_2}, the  channel matrix, $\hat{\mathbf{H}}_{ni,i}(k)$, can be written as
%\small
\begin{align}\label{noisy_channel}
\hat{\mathbf{H}}_{ni,i}(k)&=\sqrt{\Lambda_{ni,i}}\mathbf{A}_{ni,i}\mathbf{D}_{ni,i}\mathbf{B}_{ni,i}^H(k)+\mathbf{W}_i^{''}(k).
\end{align}
%\normalsize
By converting all the complex matrices to the real matrices, Unitary ESPRIT performs the computations in real, instead of complex, numbers from beginning to the end of the algorithm, and hence, reduces the computational complexity significantly. Since we are only interested in estimating UL DoAs, the noisy channel from~\eqref{noisy_channel} can be expressed as
%\small
\begin{align}\label{H_hat_k}
\hat{\mathbf{H}}_{ni,i}(k)=\mathbf{A}_{ni,i}\mathbf{S}_{ni,i}(k)+\mathbf{W}_i^{''}(k),
\end{align}
%\normalsize
where $\mathbf{S}_{ni,i}(k)=\sqrt{\Lambda_{ni,i}}\mathbf{D}_{ni,i}\mathbf{B}_{ni,i}^H(k)$.
In order to perform unitary ESPRIT, we need to use forward-backward averaging on the received signal:
\begin{align}\notag
\hat{\mathbf{H}}_{ni,i}^{fba}(k)&=\begin{bmatrix}
\hat{\mathbf{H}}_{ni,i}(k) & \mathbf{\Pi}_{N_r}\hat{\mathbf{H}}_{ni,i}^{*}(k)\mathbf{\Pi}_{N_t}
\end{bmatrix}\\\label{H_nii_fba}
&=\begin{bmatrix}
\mathbf{A}_{ni,i}\mathbf{S}_{ni,i}(k) & \mathbf{\Pi}_{N_r}\mathbf{A}_{ni,i}^{*}\mathbf{S}_{ni,i}^{*}(k)\mathbf{\Pi}_{N_t}
\end{bmatrix}+ \begin{bmatrix}
\mathbf{W}_i^{''}(k) & \mathbf{\Pi}_{N_r}{\mathbf{W}_i^{''}}^{*}(k)\mathbf{\Pi}_{N_t}
\end{bmatrix},
\end{align}
where $\mathbf{A}^*$ represents complex conjugate of $\mathbf{A}$, and $\mathbf{\Pi}_{p}$ denotes the $p \times p$ exchange matrix with ones on its antidiagonal and zeros elsewhere. The subspace decomposition of the signal space of the received signal through singular value decomposition then can be written as:
\begin{align} \label{Noiseless_SVD}
&\begin{bmatrix}
\mathbf{A}_{ni,i}\mathbf{S}_{ni,i}(k) & \mathbf{\Pi}_{N_r}\mathbf{A}_{ni,i}^{*}\mathbf{S}_{ni,i}^{*}(k)\mathbf{\Pi}_{V}
\end{bmatrix}=\begin{bmatrix}
\mathbf{U}^{sig}_{ni,i} & \mathbf{U}^{noise}_{ni,i}
\end{bmatrix}\begin{bmatrix}
\mathbf{\Sigma}^{sig}_{ni,i} & \mathbf{0}\\
\mathbf{0} & \mathbf{0}
\end{bmatrix}\begin{bmatrix}
{\mathbf{V}^{sig}_{ni,i}}^H\\
{\mathbf{V}^{noise}_{ni,i}}^H
\end{bmatrix}
\end{align}
From this step onward, we can now follow our line of work \cite{Rubayet_Journal} in order to apply ESPRIT-based techniques on \eqref{Noiseless_SVD}. Hence, the details are not repeated here due to page limitation.
\subsection{RMSE Characterization}
The theoretical performance of the unitary ESPRIT-based UL DoA estimation can be characterized where the root mean squared error (RMSE) of the estimation is served as the performance metric.
Let $\hat{v}_{ni,i,\ell}$ denote the estimated spatial frequency for $\ell$-th tap of the target user's channel, i.e, the channel between the $i$-th BS and the $n$-th MS in the $i$-th cell; the estimation error is then given by $\Delta v_{ni,i,\ell}=v_{ni,i,\ell}-\hat{v}_{ni,i,\ell}$.
Similarly, $\Delta u_{ni,i,\ell}=u_{ni,i,\ell}-\hat{u}_{ni,i,\ell}$.

It has been shown in~\cite{analytical_performance} that the unitary transformation does not affect the MSE of the ESPRIT methods; however, the statistics of the noise and the signal subspace are changed due to the forward and backward averaging performed in \eqref{H_nii_fba}.
To be specific, the covariance and complementary covariance matrices for the equivalent noise-plus-interference matrix $\mathbf{W}_{i}^{''}(k)$ in \eqref{abcd} become, respectively~\cite{analytical_performance}:
\begin{equation}
\begin{split}
&\mathbf{R}_{i}^{(fba)}(k) = \begin{bmatrix}
\mathbf{R}_{i}(k)&\mathbf{0}\\
\mathbf{0}&\bm{\Pi}_{N_rN_t}\mathbf{R}_{i}^*(k)\bm{\Pi}_{N_rN_t}
\end{bmatrix}; \ \ \mathbf{C}_{i}^{(fba)}(k) = \begin{bmatrix}
\mathbf{0}&\mathbf{R}_{i}(k)\bm{\Pi}_{N_rN_t}\\
\bm{\Pi}_{N_rN_t}\mathbf{R}_{i}^*(k)&\mathbf{0}
\end{bmatrix},
\label{equ_covariance}
\end{split}
\end{equation}
where $\mathbf{R}_{i}(k) = \mathbb{E}_{\alpha,\theta,\phi,\psi}\left\{\text{vec}\left\{\mathbf{W}_i^{''}(k)\right\}\text{vec}\left\{\mathbf{W}_i^{''}(k)\right\}^H\right\}$, and the expectation, $\mathbb{E}_{\alpha,\theta,\phi,\psi}$, is taken with respect to different channel realizations (i.e., w.r.t. channel gains, DoA's-- both azimuth and elevation-- and DoD's of the interference channels). Now the expression of the covariance matrix, $\mathbf{R}_{i}(k)$, can be simplified using the following lemma:

\begin{lemma}\label{lemma_covariance}
The covariance matrix, $\mathbf{R}_{i}(k)$, of the equivalent noise-plus-interference matrix, $\mathbf{W}_i^{''}$, is given by:
\begin{align}
\mathbf{R}_{i}(k)=\mathbf{R}_{i,1}(k)+\mathbf{R}_{i,2}(k)+\mathbf{R}_{i,3}(k)+\mathbf{R}_{i,4}(k),
\end{align}where $\mathbf{R}_{i,4}(k)=\sigma^2\mathbf{I}_{N_rN_t}$, where $\sigma^2$ is the noise variance, and
%\small
\begin{align}\label{R_i_1_k}
\mathbf{R}_{i,1}(k)=\mathds{E}_{\alpha,\theta,\phi,\psi}\left[\sum\limits_{\substack{g=0\\g\neq i}}^{G-1}\left(\sqrt{\Lambda_{ng,i}}\right)^2\bm{\mathcal{R}}_{ng,i}(k)\right],
\end{align}
\begin{align}\label{R_i_2_k}
&\mathbf{R}_{i,2}(k)=\mathds{E}_{\alpha,\theta,\phi,\psi}\left[\rho_1^2\sum\limits_{\substack{j=0\\j\neq n}}^{J-1}\left(\sqrt{\Lambda_{ji,i}}\right)^2\bm{\mathcal{X}}_{t,r}\bm{\mathcal{R}}_{ji,i}(k)\bm{\mathcal{X}}_{t,r}\right],
\end{align}
\begin{align}\label{R_i_3_k}
&\mathbf{R}_{i,3}(k)=\mathds{E}_{\alpha,\theta,\phi,\psi}\left[\rho_1^2\sum\limits_{\substack{g=0\\g \neq i}}^{G-1}\sum\limits_{\substack{j=0\\j\neq n}}^{J-1}\left(\sqrt{\Lambda_{jg,i}}\right)^2\bm{\mathcal{X}}_{t,r}\bm{\mathcal{R}}_{jg,i}(k)\bm{\mathcal{X}}_{t,r}\right],
\end{align}
%\normalsize
%\small
%\begin{align}\label{R_i_1_k}
	%&\mathbf{R}_{i,1}(k)=\mathds{E}_{\alpha,\theta,\phi,\psi}\left[\sum\limits_{\substack{g=0\\g\neq i}}^{G-1}\left(\sqrt{\Lambda_{ng,i}}\right)^2\bm{\mathcal{R}}_{ng,i}(k)\right],\\\label{R_i_2_k}
	%&\mathbf{R}_{i,2}(k)=\mathds{E}_{\alpha,\theta,\phi,\psi}\left[\rho_1^2\sum\limits_{\substack{j=0\\j\neq n}}^{J-1}\left(\sqrt{\Lambda_{ji,i}}\right)^2\left(\mathbf{1}_{N_t}\otimes\mathbf{I}_{N_r}\right)\right.\\
	%& \left. \quad \quad\quad\quad\quad\quad\quad\quad\quad\quad\quad \times \bm{\mathcal{R}}_{ji,i}(k)\left(\mathbf{1}_{N_t}\otimes\mathbf{I}_{N_r}\right)\vphantom{\rho_1^2\sum\limits_{\substack{j=0\\j\neq n}}^{J-1}\left(\sqrt{\Lambda_{ji,i}}\right)^2}\right],\\\label{R_i_3_k}
	%&\mathbf{R}_{i,3}(k)=\mathds{E}_{\alpha,\theta,\phi,\psi}\left[\rho_1^2\sum\limits_{\substack{g=0\\g \neq i}}^{G-1}\sum\limits_{\substack{j=0\\j\neq n}}^{J-1}\left(\sqrt{\Lambda_{jg,i}}\right)^2\left(\mathbf{1}_{N_t}\otimes\mathbf{I}_{N_r}\right)\right.\\
	%&\left. \quad \quad\quad\quad\quad\quad\quad\quad\quad\quad\quad \times \bm{\mathcal{R}}_{jg,i}(k)\left(\mathbf{1}_{N_t}\otimes\mathbf{I}_{N_r}\right) \vphantom{\rho_1^2\sum\limits_{\substack{j=0\\j\neq n}}^{J-1}\left(\sqrt{\Lambda_{ji,i}}\right)^2}\right],
	%\end{align}\normalsize
where\small $\bm{\mathcal{X}}_{t,r}=\left(\mathbf{1}_{N_t}\otimes\mathbf{I}_{N_r}\right)$, and $\bm{\mathcal{R}}_{pq,r}(k)=\bm{\mathcal{P}}_{pq,r}(k)\bm{\mathcal{P}}_{pq,r}^H(k)$, where   $\bm{\mathcal{P}}_{pq,r}(k)=\left(\mathbf{B}_{pq,r}^*(k)\otimes\mathbf{A}_{pq,r}\right)\text{vec}\left\{\mathbf{D}_{pq,r}\right\}$.  \normalsize
\end{lemma}
\begin{proof}[Proof Sketch]
Lemma \ref{lemma_covariance} can be proved using the properties of matrix vectorization, and with the assumption of the independence of channel gains.
\end{proof}
It is noteworthy here that in Lemma \ref{lemma_covariance}, $\mathbf{R}_{i,1}(k)$, $\mathbf{R}_{i,2}(k)$, $\mathbf{R}_{i,3}(k)$, and $\mathbf{R}_{i,4}(k)$ correspond, respectively, to the effects of pilot contamination, intra-cell interference, inter-cell interference, and noise element of the noise-plus-interference signal.
Now,  the first order approximation of the mean square estimation error of $v_{ni,i,\ell}$ for the Unitary ESPRIT is given by \cite{analytical_performance}:
%\small
\begin{align}\notag
\mathbb{E}\left\{\left(\triangle v_{ni,i,\ell} \right)^2\right\}& = \frac{1}{2}\left({\mathbf{r}_{ni,i,\ell}^{(v)^{H}}}\cdot \mathbf{W}_{ni,i,mat}^* \cdot \mathbf{R}_{i}^{(fba)^{T}} \cdot \mathbf{W}_{ni,i,mat}^T {\mathbf{r}_{ni,i,\ell}^{(v)}} \right.\\\label{equ_mse1}
&\left.
- \text{Re} \left\{{\mathbf{r}_{ni,i,\ell}^{(v)^{T}}}\cdot \mathbf{W}_{ni,i,mat} \cdot \mathbf{C}_{i}^{(fba)} \cdot \mathbf{W}_{ni,i,mat}^T \cdot {\mathbf{r}_{ni,i,\ell}^{(v)}}\right\}\right),
\end{align}
%\normalsize
where
\begin{equation}\label{r_ni_i_l}
{\mathbf{r}_{ni,i,\ell}^{(v)}}=\mathbf{q}_{\ell} \otimes \left(\left[(\tilde{\mathbf{J}}_1^{(v)}\mathbf{U}^{sig}_{ni,i})^{+}(\tilde{\mathbf{J}}_2^{(v)}/e^{jv_{ni,i,\ell}}-\tilde{\mathbf{J}}_1^{(v)})\right]^T\mathbf{p}_{\ell}\right),
\end{equation}
\begin{equation}\label{W_ni_i_mat}
\mathbf{W}_{ni,i,mat}=({\mathbf{\Sigma}^{{sig}^{-1}}_{ni,i}}\mathbf{V}^{{sig}^T}_{ni,i}) \otimes (\mathbf{U}^{noise}_{ni,i} \mathbf{U}^{{noise}^H}_{ni,i}).
\end{equation}
Here,  $\mathbf{\tilde{J}}_{v,1}=[\mathbf{I}_{M_2-1}\quad \mathbf{0}]\otimes\mathbf{I}_{M_1}$ and $\mathbf{\tilde{J}}_{v,2}=[\mathbf{0}\quad \mathbf{I}_{M_2-1}]\otimes\mathbf{I}_{M_1}$ are the selection matrices for the first and second subarrays, respectively, for the spatial frequency $v_{ni,i,\ell}$;
$\mathbf{T}$ is the transformation matrix,  $\mathbf{q}_{\ell}$ is the $\ell$-th column of matrix $\mathbf{T}_{ni,i}$, $\mathbf{p}_{\ell}^T$ is the $\ell$-th row of matrix $\mathbf{T}_{ni,i}^{-1}$; $\mathbf{R}_{i}^{fba}$ and $\mathbf{C}_{i}^{fba}$ are the covariance and complementary covariance matrices of the noise-plus-interference, respectively. Now, let us consider the following lemma:

\begin{lemma}\label{lemma_R_i_fba_decompose}
Covariance and complementary covariance matrices of the forward-backward averaged signal can be decomposed as
%\small
\begin{align}\label{R_i_fba_k}	&\mathbf{R}_{i}^{(fba)}(k)=\mathbf{R}_{i,1}^{(fba)}(k)+\mathbf{R}_{i,2}^{(fba)}(k)+\mathbf{R}_{i,3}^{(fba)}(k)+\mathbf{R}_{i,4}^{(fba)}(k),\\ &\mathbf{C}_{i}^{(fba)}(k)=\mathbf{C}_{i,1}^{(fba)}(k)+\mathbf{C}_{i,2}^{(fba)}(k)+\mathbf{C}_{i,3}^{(fba)}(k)+\mathbf{C}_{i,4}^{(fba)}(k),
\end{align}
%\normalsize
where
\begin{equation}
\begin{split}
&\mathbf{R}_{i,m}^{(fba)}(k) = \begin{bmatrix}
\mathbf{R}_{i,m}(k)&\mathbf{0}\\
\mathbf{0}&\bm{\Pi}_{N_rN_t}\mathbf{R}_{i,m}^*(k)\bm{\Pi}_{N_rN_t}
\end{bmatrix};\ \ \mathbf{C}_{i,m}^{(fba)}(k) = \begin{bmatrix}
\mathbf{0}&\mathbf{R}_{i,m}(k)\bm{\Pi}_{N_rN_t}\\
\bm{\Pi}_{N_rN_t}\mathbf{R}_{i,m}^*(k)&\mathbf{0}
\end{bmatrix},
\label{equ_covariance_1}
\end{split}
\end{equation}
for $m = 1, \ldots, 4$, where $\mathbf{R}_{i,m}(k)$'s are given by Lemma \ref{lemma_covariance}.
\end{lemma}
\begin{proof}[Proof Sketch] This Lemma can be proved by substituting \eqref{R_i_fba_k} into \eqref{equ_covariance}, and by utilizing the definitions of  $\mathbf{R}_{i,m}(k)$s  from Lemma \ref{lemma_covariance}.
\end{proof}
%\begin{proof}
%Omitted due to page limitation.
%\end{proof}
Using Lemma \ref{lemma_R_i_fba_decompose}, we can separately investigate the effects of different elements of noise-plus-interference signal on the DoA estimation performance, and hence, can write \eqref{equ_mse1} as %\small
$\mathbb{E}\left\{\left(\triangle v_{ni,i,\ell} \right)^2\right\}=\sum\limits_{m=1}^{4}\mathbb{E}\left\{\left(\triangle v_{ni,i,\ell} \right)^2\right\}_m$ where
%\small
\begin{align}\notag
\mathbb{E}\left\{\left(\triangle v_{ni,i,\ell} \right)^2\right\}_m &= \frac{1}{2}\left({\mathbf{r}_{ni,i,\ell}^{(v)^{H}}}\cdot \mathbf{W}_{ni,i,mat}^* \cdot \mathbf{R}_{i,m}^{(fba)^{T}} \cdot \mathbf{W}_{ni,i,mat}^T {\mathbf{r}_{ni,i,\ell}^{(v)}} \right.\\\label{equ_mse1_new1}
&\left.
- \text{Re} \left\{{\mathbf{r}_{ni,i,\ell}^{(v)^{T}}}\cdot \mathbf{W}_{ni,i,mat} \cdot \mathbf{C}_{i,m}^{(fba)} \cdot \mathbf{W}_{ni,i,mat}^T \cdot {\mathbf{r}_{ni,i,\ell}^{(v)}}\right\}\right).
\end{align}
%\normalsize
Now, \eqref{equ_mse1_new1} depends on the singular value decomposition (SVD) of the noiseless received signal, which can be difficult to obtain at the BS.
%In fact, it is extremely difficult to simplify such complicated result in the multiple path scenarios.
However, for massive MIMO systems, this becomes possible due to the fact that the steering vectors are orthogonal.
We consider the following Lemma~\cite{Rubayet_Journal} to facilitate the derivation of the MSE expression for massive MIMO systems:
\begin{lemma}
If the elevation and azimuth angles are both drawn independently from a continuous distribution, the normalized array response vectors become orthogonal asymptotically, that is, $\mathbf{\bar{e}}_{r,jg,i}(m)\perp\textbf{span}\left\{\mathbf{\bar{e}}_{r,{j^{'}}{g^{'}},{i^{'}}}(n)\;|\;\forall (j,g,i,m)\ne (j^{'},g^{'},i^{'},n)\right\}$ when the number of antennas at the base station goes large, where $\mathbf{\bar{e}}_{r,jg,i}(m)=\frac{1}{\sqrt{N_r}}\mathbf{e}_{r,jg,i}(m)$.
	\label{lemma:angle}
\end{lemma}
Using this property, we can analytically characterize the effect of each individual perturbation element on the DoA estimation performance.
To be specific, the MSE of UL DoA estimation due to pilot contamination is given by the following Theorem:
\begin{theorem}\label{Theorem1}
	{For the massive MIMO network, the MSE of the unitary ESPRIT-based UL DoA estimation due to pilot contamination is given by…
	}
%\small
\begin{align}\notag
\mathds{E}_{\theta,\phi,\phi}\{(\Delta v_{ni,i,\ell})^2\}_{1}=&\frac{1}{8|\alpha_{ni,i}(\ell)|^2N_t^2\Lambda_{ni,i}(M_2-1)^2M_1^2} \times \\\label{Thm1Eq}
&\qquad\sum\limits_{\substack{g=0\\g\neq i}}^{G-1}\Lambda_{ng,i} X_{ng,i} \left(\sum\limits_{m=0}^{L_{ng,i-1}}|\alpha_{ng,i}(m)|^2\right)\left(Y_{ng,i}+Y_{ng,i}^{'}-2\Re\left\{e^{j\Phi}\tilde{Y}_{ng,i}\right\}\right),
\end{align}
%\normalsize
where $\Phi=\left((M_1-1)u_{ni,i,\ell}+(M_2-1)v_{ni,i,\ell}\right)$, and  $X_{ng,i}$ and $Y_{ng,i}$ are given by
%\small
\begin{align} X_{ng,i}&=\mathds{E}_{\psi}\left|\left(1+e^{-j(\omega_{ni,i,\ell}-\omega_{ng,i,m})}+\ldots+e^{-j(N_t-1)(\omega_{ni,i,\ell}-\omega_{ng,i,m})}\right)\right|^2,\\\notag
Y_{ng,i}&=\mathds{E}_{\theta,\phi}\left|\left(1+e^{j(u_{ni,i,\ell}-u_{ng,i,m})}+\ldots+e^{j(M_1-1)(u_{ni,i,\ell}-u_{ng,i,m})}\right) \times \right.\\
&\qquad\left. \left(e^{j(M_2-1)(v_{ni,i,\ell}-v_{ng,i,m} ) }-1\right) \right|^2,
\end{align}
%\normalsize
and $Y_{ng,i}^{'}$ and $\tilde{Y}_{ng,i}$ are given by
%\small
\begin{align}\notag	Y_{ng,i}^{'}&=\mathds{E}_{\theta,\phi}\left|\left(e^{j(M_1-1)u_{ng,i,m}}+e^{ju_{ni,i,\ell}}e^{j(M_1-2)u_{ng,i,m}}+\ldots+e^{j(M_1-1)u_{ni,i,\ell}}\right)\times\right.\\
&\qquad\left.\left(e^{j(M_2-1)v_{ni,i,\ell}}-e^{j(M_2-1)v_{ng,i,m}}\right)\right|^2,\\\notag \tilde{Y}_{ng,i}&=\mathds{E}_{\theta,\phi}\bigg[\left(e^{-j(M_1-1)u_{ng,i,m}}+e^{-ju_{ni,i,\ell}}e^{-j(M_1-2)u_{ng,i,m}}+\ldots+e^{-j(M_1-1)u_{ni,i,\ell}}\right)\times\\\notag
&\qquad\left(e^{-j(M_2-1)v_{ni,i,\ell}}-e^{-j(M_2-1)v_{ng,i,m}}\right)\left(1+\ldots+e^{(M_1-1)(u_{ng,i,m}-u_{ni,i,\ell})}\right)\times\\\label{Y_ng_i_tilde}
&\qquad\left(e^{j(M_2-1)(v_{ng,i,m}-v_{ni,i,\ell})}-1\right)\bigg]
\end{align}
%\normalsize
for $m=0,\ldots L_{ng,i}-1$, and $\mathds{E}_{\psi}$ and $\mathds{E}_{\theta,\phi}$ denote, respectively, expectations with respect to DoD and DoAs of the interference channel.
\end{theorem}

\begin{proof}
See Appendix \ref{Proof_Theorem1}.
\end{proof}
{
	
\begin{remark}\label{jacobian11}
Based on Jacobian, MSEs of the elevation and azimuth angles can be obtained from the MSEs of the spatial frequencies as follows \cite{zhu2013joint}
\begin{align}
\mathds{E}_{\theta,\phi,\phi}\{(\Delta\theta_{\ell})^2\}&=\mathds{E}_{\theta,\phi,\phi}\{(\Delta u_{\ell} )^2\}\frac{1}{\pi^2 \sin^2(\theta_{\ell})},\\
\mathds{E}_{\theta,\phi,\phi}\{(\Delta\theta_{\ell})^2\}&=\frac{\mathds{E}_{\theta,\phi,\phi}\{(\Delta u_{\ell} )^2\}\cot^2(\theta_{\ell})\cot^2(\phi_{\ell})}{\pi^2 \sin^2(\theta_{\ell})}+\frac{\mathds{E}_{\theta,\phi,\phi}\{(\Delta v_{\ell} )^2\}}{\pi^2 \sin^2(\theta_{\ell}) \sin^2(\phi_{\ell})}
 \end{align}
\end{remark}
}
\begin{remark}
	The expectation expressed in \eqref{Thm1Eq} of Theorem \ref{Theorem1} is NOT taken with respect to time, rather, it is taken with respect to the DoAs or DoDs of interference users due to random locations of those interference users.
\end{remark}
Similarly, the effect of intra-cell interference, inter-cell interference, and  noise elements on the MSE performance are characterized in Theorem 2, Theorem 3, and Theorem 4, respectively:

\begin{theorem}
	{For the massive MIMO network, the MSE of the unitary ESPRIT-based UL DoA estimation due to intra-cell interference is given by
	}
%\small
\begin{align}\notag
	%\begin{split}
\mathbb{E}_{\theta,\phi,\phi}\left\{\left(\triangle v_{ni,i,\ell} \right)^2\right\}_2&=\frac{\rho_1^2{|X_{ni,i,\ell}^{''}|}^2}{8|\alpha_{ni,i}(\ell)|^2N_t^2\Lambda_{ni,i}(M_2-1)^2M_1^2}\times\\
&\qquad\sum\limits_{\substack{j=0\\j\neq n}}^{J-1}\Lambda_{ji,i}{X}_{ji,i}\left(\sum\limits_{m=0}^{L_{ji,i}-1}|\alpha_{ji,i}(m)|^2\right)\left({Y}_{ji,i}+{Y}_{ji,i}^{'}-2\Re\left\{e^{j\Phi}\tilde{Y}_{ji,i}\right\}\right),
\end{align}
%\normalsize
where $X_{ni,i,\ell}^{''}=\sum\limits_{m=0}^{N_t-1}e^{jm\omega_{ni,i,\ell}}$.
\end{theorem}
\begin{proof}
See Appendix \ref{Proof_Theorem2}.
\end{proof}

\begin{theorem}
	{For the massive MIMO network, the MSE of the unitary ESPRIT-based UL DoA estimation due to inter-cell interference is given by}
%\small
\begin{align}\notag
\mathbb{E}_{\theta,\phi,\phi}\left\{\left(\triangle v_{ni,i,\ell} \right)^2\right\}_3&=\frac{\rho_1^2{|X_{ni,i,\ell}^{''}|}^2}{8|\alpha_{ni,i}(\ell)|^2N_t^2\Lambda_{ni,i}(M_2-1)^2M_1^2}\times\\
&\qquad\sum\limits_{\substack{g=0\\g\neq i}}^{G-1}\sum\limits_{\substack{j=0\\j\neq n}}^{J-1}\Lambda_{jg,i}{X}_{jg,i}\left(\sum\limits_{m=0}^{L_{jg,i}-1}|\alpha_{jg,i}(m)|^2\right)\left({Y}_{jg,i}+{Y}_{jg,i}^{'}-2\Re\left\{e^{j\Phi}\tilde{Y}_{jg,i}\right\}\right),
\end{align}
%\normalsize
\end{theorem}
\begin{proof}
The proof is similar to the proof of Theorem 2.
\end{proof}
\begin{theorem}
	{For the massive MIMO network, the MSE of the unitary ESPRIT-based UL DoA estimation due to noise element is given by
	}
\begin{align}\notag
%\begin{split}
&\mathbb{E}\left\{\left(\triangle v_{ni,i,\ell} \right)^2\right\}_4=\frac{\sigma^2}{2|\alpha_{ni,i}(\ell)|^2N_t\Lambda_{ni,i}(M_2-1)^2M_1},
\end{align}
where $\sigma^2$ is the noise variance.
\end{theorem}
\begin{proof}
This theorem can be proved following the line of proof for Theorem-1 in \cite{Rubayet_Journal}.
\end{proof}
Similarly,  we can also obtain the MSE expressions for elevation spatial frequency, $\mathbb{E}\left\{\left(\triangle u_{ni,i,\ell} \right)^2\right\}$.
Accordingly, based on Jacobian matrices, we can characterize MSE expressions for UL elevation and azimuth DoAs from the MSEs of the spatial frequencies.
%\begin{remark}
%	Theorems 1-4 show that proposed ESPRIT-based DoA estimation performance is directly related to the channel gain of the respective path, i.e., if the channel gain increases, the DoA estimation performance also gets improved.  If there is one or multiple paths overlapped with target user's path, the overall channel gain of that path increases, and hence estimation accuracy also gets better.
%\end{remark}
\begin{remark}
From Theorem 2 and Theorem 3 we can observe that non-zero correlation among the spreading sequences of different MSs does cause intra- and inter-cell interference for UL DoA estimation, and the corresponding MSEs of the estimation are directly affected by the correlation coefficient, $\rho_1$.
On the other hand, as we can see from Theorem 1, the MSE due to pilot contamination is not dependent on the correlation coefficient.
%{\color{blue}\sout{Therefore, the impact of pilot contamination on the performance of UL DoA estimation is expected to be more significant than those caused by multi-cell interference due to non-orthogonality among spreading sequences due to the fact that $\rho_1$ will be relatively small.}}

Furthermore, these four theorems suggest that our original work in~\cite{Rubayet_Journal} may yield strictly suboptimal solutions since in that work we only consider DoA estimation error due to noise elements.
This observation will be verified through performance evaluation in Section V.
\end{remark}
{
\subsection{Complexity Analysis}\label{CompComplexity_DoA}
In this subsection, we discuss the computational coomplexity of unitary ESPRIT-based DoA estimation procedure as presented in Section \ref{UL_DOA_Estimation}. We first summarize the computational complexity of some basic operations in terms of floating point operations (FLOPS). It requires $2(n-1)mp$ FLOPS for computing product of two matrices of sizes $(m \times n)$ and $n \times p$. For taking inverse of a positive definite matrix of size $(n \times n)$ requires $(n^3 +n^2+n)$ floating point operations; number of FLOPS required for taking SVD of an $m \times n$ matrix is $(4m^2n+22n^3)$, and complexity for finding eigenvalues of an $n \times n$ matrix is $n^3$. Now, we can describe the computational complexity of each step of our algorithm presented. For complexity analysis, we assume all the channels have $L$ resolvable paths. Now, correlating the received signal with training symbol matrix in  \eqref{Z_i_k_2} requires $C_a=2(Q-1)N_rN_t$ number of FLOPS. Taking forward-backward averaging in \eqref{H_nii_fba} requires $C_b=2N_rN_t(N_r+N_t-2)$ FLOPS. Number of FLOPS required for taking SVD of the forward-backward averaged received signal in \eqref{Noiseless_SVD} is $C_c=(8N_r^2N_t +176N_t^3)$. Now, for solving shift-invariance equations for  elevation and azimuth spatial frequencies, the number of FLOPS required are, respectively, $C_d=2[M_2(M_1-1)-1]L^2+[L^3+L^2+L]+[2(L-1)M_2L(M_1-1)]$, and $C_e=2[M_1(M_2-1)-1]L^2+[L^3+L^2+L]+[2(L-1)M_1L(M_2-1)]$. Finally, for calculating the eigenvalues of two shift-invariance operator matrices requires $C_f=2L^3$ number of FLOPS. Hence, total computational complexity of our ESPRIT-based DoA estimation method can be written as $C_{\text{ESPRIT}}=C_a+C_b+C_c+C_d+C_e+C_f$. Next, for comparison. we compute the computational complexity of MUSIC algorithm. For computing the covariance matrix of the received signal, the number of FLOPS required is $D_a=(Q+1)N_r^2$. Next, computing the SVD of the covariance matrix requires $D_b=26N_r^3$ FLOPS. Let $N_g$ denote the number of grids for candidate DoA search. Hence, total number of FLOPS required for extracting the eigenvectors corresponding to noise subspace is $D_c=N_g\left([2N_r(N_r-L)]+[2N_r-3]\right)$. Hence computational cost for MUSIC algorithm is $C_{\text{MUSIC}}=D_a+D_b+D_c$.

}
\section{Downlink Precoding and Achievable Rate Analysis}
\subsection{Optimum Precoding for Sum-rate Maximization}\label{Optim_Precoding}
In the DL, at the $i$-th BS, the $N_s \times 1 $ information symbol vector intended for the $n$-th MS in the $i$-th cell on the $k$-th subcarrier can be expressed as $\mathbf{s}^{dl}_{ni}[k]=\begin{bmatrix}
s^{dl}_{ni,0}[k],
\ldots,
s^{dl}_{ni,N_s-1}[k]
\end{bmatrix}$,
where $s^{dl}_{ni,p}[k]$ is the $p$-th information symbol intended for the $n$-th MS.
Accordingly, the $N_r \times 1$ downlink frequency domain transmit signal from the $i$-th BS can be written as
\begin{align}
\mathbf{x}^{dl}_{i}[k]&=\sum_{j=0}^{J-1}\mathbf{x}^{dl}_{ji}[k]=\sum_{j=0}^{J-1}\mathbf{V}_{ji}[k]\mathbf{s}^{dl}_{ji}[k],
\end{align}
where $\mathbf{x}^{dl}_{ji}[k]=\mathbf{V}_{ji}[k]\mathbf{s}^{dl}_{ji}[k]$, and $\mathbf{V}_{ji}[k]$ is the $N_r \times N_s$ precoding matrix for the $j$-th MS in the $i$-th cell on the $k$-th subcarrier.
Now, the $N_t \times 1$  received signal at the $n$-th MS in the $i$-th cell on the $k$-th subcarrier, $\mathbf{y}_{ni}^{dl}[k]$, can be written as
%\small
\begin{align}\notag
\mathbf{y}_{ni}^{dl}[k]&=\sum\limits_{g=0}^{G-1}\sqrt{\Lambda_{ni,g}}\mathbf{H}_{ni,g}^{dl}[k]\mathbf{x}_{g}^{dl}[k]+\mathbf{n}_{ni}^{dl}[k]=\sum\limits_{g=0}^{G-1}\sum_{j=0}^{J-1}\sqrt{\Lambda_{ni,g}}\mathbf{H}_{ni,g}^{dl}[k]\mathbf{V}_{jg}[k]\mathbf{s}^{dl}_{jg}[k]+\mathbf{n}_{ni}^{dl}[k]\\\notag
&=\sqrt{\Lambda_{ni,i}}\mathbf{H}_{ni,i}^{dl}[k]\mathbf{V}_{ni}[k]\mathbf{s}^{dl}_{ni}[k]+\sum\limits_{\substack{j=0\\j\neq n}}^{J-1}\sqrt{\Lambda_{ni,i}}\mathbf{H}_{ni,i}^{dl}[k]\mathbf{V}_{ji}[k]\mathbf{s}^{dl}_{ji}[k]\\\label{y_ni_dl_k}
&\quad+\sum\limits_{\substack{g=0\\g\neq i}}^{G-1}\sum\limits_{\substack{j=0\\j\neq n}}^{J-1}\sqrt{\Lambda_{ni,g}}\mathbf{H}_{ni,g}^{dl}[k]\mathbf{V}_{jg}[k]\mathbf{s}^{dl}_{jg}[k]+\mathbf{n}_{ni}^{dl}[k],
\end{align}
%\normalsize
where $\mathbf{H}_{ni,g}^{dl}[k]$ is the $N_t \times N_r$ downlink channel between the $g$-th BS and the $n$-th MS in $i$-th cell on the $k$-th subcarrier, and $\mathbf{n}_{ni}^{dl}[k]$ is the corresponding $N_t \times 1$ noise vector at the receiver with $\mathds{E}\{\mathbf{n}_{ni}^{dl}[m]\mathbf{n}_{ni}^{dl}[n]\}=\sigma^2\mathbf{I}_{N_t}\delta(m-n)$.
In \eqref{y_ni_dl_k}, the first term is the desired signal, while the second and third terms represent the intra- and inter-cell interferences, respectively.
Now, the rate for $n$-th MS in $i$-th cell is given by
%\small
\begin{align}\notag
\mathcal{I}_{ni}[k]=\log_2\det&\left(\mathbf{I}+\Lambda_{ni,i}\mathbf{H}_{ni,i}^{dl}[k]\mathbf{V}_{ni}[k]\mathbf{V}_{ni}^{H}[k]{\mathbf{H}_{ni,i}^{{dl}^{H}}}[k]\times\right.\\
&\quad\left.\left(\sum_{(n,i)\ne(j,g)}\Lambda_{ni,g}\mathbf{H}_{ni,g}^{dl}[k]\mathbf{V}_{jg}[k]\mathbf{V}_{jg}^{H}[k]{\mathbf{H}_{ni,g}^{{dl}^{H}}}[k]+\sigma^2\mathbf{I}\right)^{-1}\right).
\end{align}
%\normalsize
Accordingly, the sum-rate maximization (SRM) problem can be expressed as
\begin{align}\notag
&\max_{\{\mathbf{V}_{ji}[k]\}}  \ \ \ \sum_{j=0}^{J-1}\mathcal{I}_{ji}[k]\\\label{SRM_1}
& s.t. \ \ \ \sum_{j=0}^{J-1}\text{Tr}\left(\mathbf{V}_{ji}[k]\mathbf{V}_{ji}[k]^H \right)\leq P_t,
\end{align}
where $P_t$ is the total power available at the BS for each sucarrier.
In general, it is challenging to solve the problem in \eqref{SRM_1} since it is highly non-convex.
Alternatively, sum-MSE (mean square error) minimization is another popular utility maximization problem for DL multi-user MIMO systems.
Let $\mathbf{T}_{ni}$ be the DL receive processing matrix for the $n$-th MS.
The estimated received symbol vector can then be written as $\hat{\mathbf{s}}^{dl}_{ni}[k]=\mathbf{T}_{ni}^H\mathbf{y}^{dl}_{ni}[k]$.
Now,  $n$-th MS's MSE matrix can be defined as
\begin{align}\notag
\mathbf{E}_{ni}[k]&=\mathds{E}\left[\left(\hat{\mathbf{s}}^{dl}_{ni}[k]-\mathbf{s}^{dl}_{ni}[k]\right)\left(\hat{\mathbf{s}}^{dl}_{ni}[k]-\mathbf{s}^{dl}_{ni}[k]\right)^H\right]\\\notag
&=\left(\mathbf{I}-\sqrt{\Lambda_{ni,i}}\mathbf{T}_{ni}^H\mathbf{H}_{ni,i}^{dl}[k]\mathbf{V}_{ni}[k]\right)\left(\mathbf{I}-\sqrt{\Lambda_{ni,i}}\mathbf{T}_{ni}^H\mathbf{H}_{ni,i}^{dl}[k]\mathbf{V}_{ni}[k]\right)^H\\
& \ \ \ + \sum_{(n,i)\ne(j,g)} \Lambda_{ni,g} \mathbf{T}_{ni}^H\mathbf{H}_{ni,g}^{dl}[k]\mathbf{V}_{jg}[k]\mathbf{V}_{jg}^{H}[k]{\mathbf{H}_{ni,g}^{{dl}^{H}}}[k]\mathbf{T}_{ni}+\sigma^2\mathbf{T}_{ni}^H\mathbf{R}_{ni}
\end{align}
Accordingly, the sum-MSE minimization problem can be defined as
\begin{align}\notag
&\min_{\{\mathbf{V}_{ji}[k]\}} \ \ \ \sum_{j=0}^{J-1}\epsilon_{ji}[k]\\\label{MSE_problem}
& s.t. \ \ \ \sum_{j=0}^{J-1}\text{Tr}\left(\mathbf{V}_{ji}[k]\mathbf{V}_{ji}[k]^H \right)\leq P_t.
\end{align}
where $\epsilon_{ni}[k]=\text{Tr}\{\mathbf{E}_{ni}[k]\}$.
The relationship between the problems in \eqref{SRM_1} and \eqref{MSE_problem} can be established by the following lemma~\cite{Iteratively_Weighted_MMSE}:
\begin{lemma}
The sum-rate maximization problem in \eqref{SRM_1} and the sum-MSE minimization problem in \eqref{MSE_problem} are equivalent in the sense that the optimal solutions, $\{\mathbf{V}_{ji}[k]\}_{j=0}^{J-1}$, for both problems are identical.
\end{lemma}
In this work, we assume that no coordination is available among BSs, which is a typical scenario in TDD-based FD-MIMO networks.
Hence, problem in \eqref{MSE_problem} can be written as
%\small
\begin{align}\notag
&\min_{\{\mathbf{V}_{ji}[k]\}} \ \ \ \left|\left|\mathbf{T}_{i}^H\mathbf{H}_{i,i}^{dl}[k]\mathbf{V}_{i}[k]-\mathbf{I}\right|\right|_F^2\\\label{MSE_problem2}
& s.t. \ \ \ \sum_{j=0}^{J-1}\text{Tr}\left(\mathbf{V}_{ji}[k]\mathbf{V}_{ji}[k]^H \right)\leq P_t,
\end{align}
%\normalsize
where $\mathbf{T}_i^H=\text{blkdiag}\{\mathbf{T}_{0i}^H, \mathbf{T}_{1i}^H, \ldots, \mathbf{T}_{(J-1)i}^H,\}$, $\mathbf{V}_i[k]=\left[\mathbf{V}_{0i}[k], \mathbf{V}_{1i}[k], \ldots, \mathbf{V}_{(J-1)i}[k]\right]$.
and $\mathbf{H}_{i,i}^{dl}[k]=\left[ \sqrt{\Lambda_{0i,i}}\mathbf{H}_{0i,i}^{{dl}^T}[k], \sqrt{\Lambda_{1i,i}}\mathbf{H}_{1i,i}^{{dl}^T}[k], \ldots, \sqrt{\Lambda_{(J-1)i,i}}\mathbf{H}_{(J-1)i,i}^{{dl}^T}[k]\right]^T$.
Now, using channel reciprocity property, the downlink channel can be written in terms of the uplink channel:
%\small
\begin{align}\label{H_ji_g_dl_k}
\mathbf{H}_{ni,i}^{dl}[k]= \mathbf{H}_{ni,i}^{T}[k]=\mathbf{B}_{ni,i}^{*}[k]\mathbf{D}_{ni,i}\mathbf{A}_{ni,i}^{T}=\bar{\mathbf{B}}_{ni,i}^{*}\mathbf{D}_{ni,i}\bar{\mathbf{A}}_{ni,i}^{T}[k],
\end{align}
%\normalsize
where $\bar{\mathbf{A}}_{ni,i}[k]=\begin{bmatrix}
\mathbf{e}_{r,ni,i,k}(0) & \ldots & \mathbf{e}_{r,ni,i,k}(L_{ni,i}-1)
\end{bmatrix}$, where $\mathbf{e}_{r,ni,i,k}(\ell)=\mathbf{e}_{r,ni,i}(\ell) e^{\frac{-j2\pi k \ell}{N_c}}$,  and $\bar{\mathbf{B}}_{ni,i}=\begin{bmatrix}
\mathbf{e}_{t,ni,i}(0) & \ldots & \mathbf{e}_{t,ni,i}(L_{ni,i}-1)
\end{bmatrix}$.
Assuming each MS will only use its own DL CSI for receive processing, we have $\mathbf{T}_{n,i}^H\mathbf{H}_{ni,i}^{dl}[k]=\mathbf{D}_{ni,i}\bar{\mathbf{A}}_{ni,i}^{T}$.
Accordingly, the problem in \eqref{MSE_problem2}  can be expressed as
%\small
\begin{align}\notag
&\min_{\{\mathbf{V}_{ji}[k]\}} \ \ \ \left|\left|\bar{\mathbf{D}}_{i,i}\bar{\mathbf{A}}_{i,i}^{T}[k]\mathbf{V}_{i}[k]-\mathbf{I}\right|\right|_F^2\\\label{MSE_problem3}
& s.t. \ \ \ \sum_{j=0}^{J-1}\text{Tr}\left(\mathbf{V}_{ji}[k]\mathbf{V}_{ji}[k]^H \right)\leq P_t,
\end{align}
%\normalsize
where $\bar{\mathbf{D}}_{i,i}=\text{blkdiag}\{\bar{\mathbf{D}}_{0i,i}, \bar{\mathbf{D}}_{1i,i}, \ldots, \bar{\mathbf{D}}_{(J-1)i,i},  \}$, $\bar{\mathbf{A}}_{i,i}[k]=\left[\bar{\mathbf{A}}_{0i,i}[k], \bar{\mathbf{A}}_{1i,i}[k], \ldots, \bar{\mathbf{A}}_{(J-1)i,i}[k]\right]$, and $\bar{\mathbf{D}}_{ni,i}=\sqrt{\Lambda_{ni,i}}\mathbf{D}_{ni,i}$ accounts for both the large and small scale fading effect.
The solution to this problem is given by the following theorem:
\begin{theorem}\label{Theorem_precoder}
Let $\bar{\mathbf{D}}_{i,i}\bar{\mathbf{A}}_{i,i}^{T}[k]=\tilde{\mathbf{U}}_{i,i}[\tilde{\mathbf{\Lambda}}_{i,i}, \mathbf{0}]\tilde{\mathbf{W}}_{i,i}^H$  be the SVD of the effective channel, $\bar{\mathbf{D}}_{i,i}\bar{\mathbf{A}}_{i,i}^{T}[k]$, where $\tilde{\mathbf{\Lambda}}_{i,i}=\text{diag}\{\lambda_{0i,i}, \lambda_{1i,i}, \ldots, \lambda_{(JN_s-1)i,i}\}$. Then the optimal precoding matrix problem in \eqref{MSE_problem3} is given by $\mathbf{V}_i[k]=\tilde{\mathbf{W}}_{i,i} [\mathbf{\Xi}_{i,i},\mathbf{0}]^T\tilde{\mathbf{U}}_{i,i}^H$, where $\mathbf{\Xi}_{i,i}=\text{diag}\{\xi_{0i,i}, \xi_{1i,i}, \ldots, \xi_{(JN_s-1)i,i}\}$, and $\xi_{mi,i}=\lambda_{mi,i}/(\lambda_{mi,i}^2+\eta)$, with the smallest $\eta\ge 0$ satisfying $\sum_{m=0}^{JN_s-1}|\xi_{mi,i}|^2\le P_t$.
\end{theorem}
\begin{proof}
See Appendix \ref{Proof_Theorem5}.
\end{proof}
From theorem \ref{Theorem_precoder}, it can be seen that the optimal precoder that minimizes the sum-MSE, and hence maximizes the sum-rate can be constructed from the estimated UL DoAs as well as the path gains.
In this paper, we assume that the BS has perfect knowledge of the path gains.
However, path gains can also be estimated using  maximum likelihood (ML) method once the DoAs have been estimated.
Our work on this aspect can be found in~\cite{Joint_parametric_ICC17}.

\subsection{Large-Antenna System Analysis}
In this section, we present the achievable rate analysis and simplified precoding strategy for massive FD-MIMO systems.
{Our discussions in this sub-section are based on asymptotic analysis. This can be viewed as the special case of Section \ref{Optim_Precoding} where the number of antennas at the base station goes large asymptotically.
}
\subsubsection{Achievable Rate under Perfect Channel Estimation}
In this case, \eqref{y_ni_dl_k} can be written as
%\small
\begin{align}\label{y_ni_dl_k_new2}
&\mathbf{y}_{ni}^{dl}[k]=\bar{\mathbf{B}}_{ni,i}^{*}\bar{\mathbf{D}}_{ni,i}\bar{\mathbf{A}}_{ni,i}^{T}[k]\mathbf{V}_{ni}[k]\mathbf{s}^{dl}_{ni}[k]+\sum\limits_{\substack{j=0\\j\neq n}}^{J-1}\bar{\mathbf{B}}_{ni,i}^{*}\bar{\mathbf{D}}_{ni,i}\bar{\mathbf{A}}_{ni,i}^{T}[k]\mathbf{V}_{ji}[k]\mathbf{s}^{dl}_{ji}[k]+\mathbf{n}_{ni}^{'}[k],
\end{align}
%\normalsize
where
%\small
\begin{align}
\mathbf{n}_{ni}^{'}[k]=\sum\limits_{\substack{g=0\\g\neq i}}^{G-1}\sum\limits_{\substack{j=0}}^{J-1}\sqrt{\Lambda_{ni,g}}\mathbf{H}_{ni,g}^{dl}[k]\mathbf{V}_{jg}[k]\mathbf{s}^{dl}_{jg}[k]+\mathbf{n}_{ni}^{dl}[k]
\end{align}
%\normalsize
is the equivalent noise-plus-inter-cell-interference vector.
As the number of antennas grows large, the right singular matrix, $\tilde{\mathbf{W}}_{i,i}$, in Theorem \ref{Theorem_precoder} can be approximated as the DoA  matrix, $ \bar{\mathbf{A}}_{i,i}^*$.
In other words, for massive FD-MIMO systems, eigen directions align with the directions of arrivals, which is also validated in \cite{Rubayet_Journal} and \cite{Angle_and_Delay_Estimation_for_3D}.
From Lemma \ref{lemma:angle}, the array steering vectors for different MSs become orthogonal as the number of antennas grows large, i.e., we have $(1/N_r)\bar{\mathbf{A}}_{ji,i}^{H}[k]\bar{\mathbf{A}}_{{j^{'}}i,i}[k]\rightarrow \mathbf{0}$ as $N_r \rightarrow \infty$ for all $j \ne j^{'}$.
Hence for the massive MIMO systems, beamforming in the DoA directions nullifies the intra-cell interferences. Therefore, the optimum eigen-beamformer under perfect DoA estimation:
%\small
\begin{align}\label{V_ni_i_eig_opt}
\mathbf{V}_{ni}^{eig}[k]=\frac{1}{N_r}\bar{\mathbf{A}}_{ni,i}^{*}[k],
\end{align}
%\normalsize
and accordingly, the received signal in \eqref{y_ni_dl_k_new2} can be written as
%\small
\begin{align}\label{y_ni_dl_k_new3}
\mathbf{y}_{ni}^{dl}[k]=\bar{\mathbf{B}}_{ni,i}^{*}\bar{\mathbf{D}}_{ni,i}\mathbf{s}^{dl}_{ni}[k]+\mathbf{n}_{ni}^{'}[k].
\end{align}
%\normalsize
Now, the  signal in \eqref{y_ni_dl_k_new3}, under the optimal receive processing, results in
%\small
\begin{align}\label{y_ni_dl_k_tilde}
\tilde{\mathbf{y}}_{ni}^{dl}[k]=\bar{\mathbf{D}}_{ni,i}\mathbf{s}^{dl}_{ni}[k]+\tilde{\mathbf{n}}_{ni}^{'}[k],
\end{align}
%\normalsize
where $\tilde{\mathbf{y}}_{ni}^{dl}[k]=\left(\bar{\mathbf{B}}_{ni,i}^{T}\bar{\mathbf{B}}_{ni,i}^{*}\right)^{-1}\bar{\mathbf{B}}_{ni,i}^{T}\mathbf{y}_{ni}^{dl}[k]$, and $\tilde{\mathbf{n}}_{ni}^{'}[k]=\left(\bar{\mathbf{B}}_{ni,i}^{T}\bar{\mathbf{B}}_{ni,i}^{*}\right)^{-1}\bar{\mathbf{B}}_{ni,i}^{T}\mathbf{n}_{ni}^{'}[k]$. Accordingly, the achievable rate for the $n$-th user in $i$-th cell, $\mathcal{I}_{ni}[k]$, can be expressed as
{
%\small
\begin{align}\label{R_ni_k_new1}
\mathcal{I}_{ni}[k]=\log_2\det\left(\mathbf{I}_{L_{ni,i}}+\frac{\bar{\mathbf{D}}_{ni,i}\mathbf{Q}_{ni}^{dl}[k]\bar{\mathbf{D}}_{ni,i}^{H}}{\sum\limits_{\substack{g=0\\g\neq i}}^{G-1}\sum\limits_{\substack{j=0}}^{J-1}\Lambda_{ni,g} \tilde{\mathbf{B}}_{ni,i}^H\mathbf{H}_{ni,g}^{dl}[k]\mathbf{V}_{jg}[k]\mathbf{Q}^{dl}_{jg}[k] \mathbf{V}_{jg}^H[k] \mathbf{H}_{ni,g}^{{dl}^H}[k]\tilde{\mathbf{B}}_{ni,i} +\sigma^2\mathbf{I} }\right).
\end{align}
%\normalsize
where,  $\tilde{\mathbf{B}}_{ni,i}^H= \left(\bar{\mathbf{B}}_{ni,i}^{T}\bar{\mathbf{B}}_{ni,i}^{*}\right)^{-1}\bar{\mathbf{B}}_{ni,i}^{T}$, and $\mathbf{Q}_{ni}^{dl}[k]=\mathds{E}\{\mathbf{s}^{dl}_{ni}[k]\mathbf{s}^{{dl}^{H}}_{ni}[k]\}$ is the covariance matrix of the transmit symbol vector  from the $i$-th BS intended for the $n$-th MS on the $k$-th subcarrier. 
{
Now, \eqref{R_ni_k_new1} can succinctly be written as
%{\color{blue}
%	In a cellular network, the location of transmitting MSs are usually random. Furthermore, the transmitting MSs within a cell are scheduled by the corresponding BS through different scheduling algorithms. As a result, DoAs corresponding to these inter-cell interference users are also random. Therefore, the inter-cell interference experienced by target terminals in \eqref{R_ni_k_new1} is NOT fixed and the underlying covariance matrix of interference plus noise is random. In general, it is challenging to specify an achievable rate expression for a cellular network that has a random network topology with random inter-cell interference. In this paper, we adopted a similar approach that is widely used in the field of stochastic geometry analysis of cellular networks \cite{andrews2011tractable, baccelli2009stochastic} to spatially average the inter-cell interference with respect to the random DoA and DoDs resulted from the random locations of inter-cell interference users. Accordingly, the average rate expression can be written as
%}
%\small
\begin{align}\label{I_ni_k_1}
\mathcal{I}_{ni}[k]=\log_2\det\left(\mathbf{I}_{L_{ni,i}}+\bar{\mathbf{D}}_{ni,i}\mathbf{Q}_{ni}^{dl}[k]\bar{\mathbf{D}}_{ni,i}^{H}\tilde{\mathbf{R}}_{ni}^{{'}^{-1}}[k]\right),
\end{align}
%\normalsize
where inter-cell interference-plus-noise covariance matrix,	$\tilde{\mathbf{R}}_{ni}^{{'}}[k]$, is defined as 
%\small
\begin{align}\label{R_ni_tilde_prime_k}
\tilde{\mathbf{R}}_{ni}^{{'}}[k]= \sum\limits_{\substack{g=0\\g\neq i}}^{G-1}\sum\limits_{\substack{j=0}}^{J-1}\Lambda_{ni,g} \tilde{\mathbf{B}}_{ni,i}^H\mathbf{H}_{ni,g}^{dl}[k]\mathbf{V}_{jg}[k]\mathbf{Q}^{dl}_{jg}[k] \mathbf{V}_{jg}^H[k] \mathbf{H}_{ni,g}^{{dl}^H}[k]\tilde{\mathbf{B}}_{ni,i} +\sigma^2\mathbf{I} .
\end{align}
%\normalsize
%Accordingly, the system sum mutual information, $\mathcal{I}_{i}[k]$, can be expressed as $\mathcal{I}_{i}[k]=\sum\limits_{j=0}^{J-1}\mathcal{I}_{ji}[k]$.
Let us now consider the following lemma:
\begin{lemma}\label{R_ni_tilde_k}
	Assuming all BSs apply the same precoding strategy, equivalent inter-cell interference-plus-noise covariance matrix, $\tilde{\mathbf{R}}_{ni}^{{'}}[k]$, is approximated by
	%\small
	\begin{align}\label{R_ni_tilde_prime_k_2}
	\tilde{\mathbf{R}}_{ni}^{{'}}[k]\approx (\zeta_{ni}+\sigma^2)\mathbf{I}_{L_{ni,i}},
	\end{align}
	%\normalsize
	where $\zeta_{ni}=J(G-1)\mathds{E}\{\Lambda_{ni,g}p_{jg,\ell}[k]|\alpha_{ni,g}(\ell)|^2\}$, where $p_{jg,\ell}[k]$ is the power allocated on the $\ell$-th symbol for $j$-th user in $g$-th cell on the $k$-th subcarrier.
\end{lemma}
\begin{proof}
This lemma can be proved by substituting \eqref{H_ji_g_dl_k} in \eqref{R_ni_tilde_prime_k}, and by utilizing the orthogonality property from Lemma \ref{lemma:angle}. Details are omitted due to page limitation.
\end{proof}
%Note that due to the presence of large scale fading parameter for inter-cell users, in Lemma~\ref{R_ni_tilde_k}, $\zeta_{ni} << \sigma^2$.
%As a result, diagonal elements of $\tilde{\mathbf{R}}_{ni}^{'}[k]$ is much larger than the non-diagonal elements.
%%This dominance is even  stronger when the covariance matrix is inverted.
%This means we can approximate the inverse of the covariance matrix as $\tilde{\mathbf{R}}_{ni}^{{'}^{-1}}[k]\approx(\frac{1}{\zeta_{ni}+\sigma^2})\mathbf{I}$. 
Accordingly, \eqref{I_ni_k_1} results in
%\small
\begin{align}\label{I_ni_k_2}
\mathcal{I}_{ni}[k]=\log_2\det\left(\mathbf{I}_{L_{ni,i}}+\frac{1}{\zeta_{ni}+\sigma^2}\bar{\mathbf{D}}_{ni,i}\mathbf{Q}_{ni}^{dl}[k]\bar{\mathbf{D}}_{ni,i}^{H}\right).
\end{align}
%\normalsize
}
Assuming Gaussian input signal, $\mathbf{Q}_{ni}^{dl}[k]=\mathds{E}\{\mathbf{s}^{dl}_{ni}[k]\mathbf{s}^{{dl}^{H}}_{ni}[k]\}=\text{diag}\{p_{ni,0}[k],\ldots, p_{ni,L-1}[k] \}$, where  $p_{ni,\ell}[k]$ is the power to be allocated on the $\ell$-th information symbol on the $k$-th subcarrier for the target user.
Now, using Hadamard inequality, \eqref{I_ni_k_2} can be rewritten as
%\small
\begin{align}\label{I_ni_k_3}
\mathcal{I}_{ni}[k]=\log_2 \underset{\ell}{\Pi} \left(1+\frac{\Lambda_{ni,i}|\alpha_{ni,i}(\ell)|^2p_{ni,\ell}[k]}{\zeta_{ni}+\sigma^2}\right)=\sum\limits_{\ell=0}^{L_{ni,i}-1}\log_2\left(1+\gamma_{ni,\ell}p_{ni,\ell}[k]\right),
\end{align}
%\normalsize
where $\gamma_{ni,\ell}=\Lambda_{ni,i}|\alpha_{ni,i}(\ell)|^2/(\zeta_{ni}+\sigma^2)$. Accordingly, the optimal power allocation under perfect DoA estimation is the well-known water-filling solution which can be expressed as
\begin{align}\label{traditional_water_filling}
p_{ni,\ell}[k]=[\mu_{ni,\ell}[k]-1/\gamma_{ni,\ell}]^{\Diamond},
\end{align}
%$p_{ni,\ell}[k]=[\mu_{ni,\ell}[k]-1/\gamma_{ni,\ell}]^{\Diamond}$,
where $[x]^{\Diamond}$ denotes a function with $[x]^{\Diamond}=0$ when $x<0$, and $[x]^{\Diamond}=x$ when $x>0$, and $\mu_{ni,\ell}[k]$ is the corresponding Lagrange multiplier.
%\begin{remark}
%It can be shown that for the massive MIMO systems, the array steering matrices containing the DoAs become asymptotically orthogonal \cite{Rubayet_Journal}, i.e., as $N_r \rightarrow \infty$, $\frac{1}{N_r}\mathbf{A}_{ni,g}^{T}[k]\mathbf{A}_{jg,g}^{*}[k]\approx \mathbf{0},\forall \{n,i\}\neq\{j,g\}$. Therefore, \eqref{R_ni_tilde_prime_k} results in $\tilde{\mathbf{R}}_{ni}^{{'}}[k] \rightarrow \sigma^2\mathbf{I}_L$. In other words, the effects of inter-cell interference is minimized as the number of BS antennas goes large. Hence, for the large antenna systems, the mutual information in \eqref{I_ni_k_3} can be expressed as
%\begin{align}\label{I_ni_k_prime}
%\mathcal{I}_{ni}^{'}[k]=\sum\limits_{\ell=0}^{L-1}\log_2\left(1+\gamma_{ni,\ell}^{'}p_{ni,\ell}^{'}[k]\right),
%\end{align}
%where $\gamma_{ni,\ell}^{'}=\Lambda_{ni,i}|\alpha_{ni,i}(\ell)|^2/\sigma^2$, and $p_{ni,\ell}^{'}[k]$ is the corresponding water-filling power allocation.
%\end{remark}

\subsubsection{System Achievable Rate under DoA Estimation Errors}
\label{subsec:DoAerror}
%In this section, we study the effects of DoA estimation error on the achievable rate of the 3D massive-MIMO OFDM system.
In this case, since the BS does not have perfect DoA estimation, the array steering matrix for the $n$-th MS in $i$-th cell, in the presence of DoA estimation error, can be expressed in the form of
%\small
\begin{equation*}
\hat{\bar{\mathbf{A}}}_{ni,i}[k]=\begin{bmatrix}
\hat{\mathbf{e}}_{r,ni,i,k}(0) & \hat{\mathbf{e}}_{r,ni,i,k}(1) & \ldots & \hat{\mathbf{e}}_{r,ni,i,k}(L_{ni,i}-1)
\end{bmatrix},
\end{equation*}
%\normalsize
where $\hat{\mathbf{e}}_{r,ni,i,k}(\ell)= e^{\frac{-j2\pi k \ell}{N_c}}\mathbf{a}(v_{ni,i,\ell}+\Delta v_{ni,i,\ell})\otimes \mathbf{a}(u_{ni,i,\ell}+\Delta u_{ni,i,\ell})$, and $\Delta u_{ni,i,\ell}$ and $\Delta v_{ni,i,\ell}$ represent the DoA estimation errors in the azimuth and elevation spatial frequencies for the $\ell$-th path of the channel between the $i$-th BS and the $n$-th MS in $i$-th cell.
Now, let us consider the following lemma:
\begin{lemma}\label{error_orthogonality_lemma}
For the massive FD-MIMO OFDM system, the  normalized steering vectors $\mathbf{\bar{e}}_{r,jg,i,k}(\ell)=1/\sqrt{N_r}\mathbf{e}_{r,jg,i,k}(\ell)$ and $\mathbf{\hat{\bar{e}}}_{r,j'g',i,k}(\ell')=1/\sqrt{N_r}\mathbf{\hat{e}}_{r,j'g',i',k}(\ell')$, $\forall \{j,g,i,\ell\}\ne \{j',g',i',\ell'\}$, becomes orthonormal asymptotically as the number of antenna, $N_r \rightarrow \infty$ .
\end{lemma}
\begin{proof}
A similar lemma is proved in \cite{Rubayet_Journal}, and hence omitted here.
\end{proof}
Using Lemma \ref{error_orthogonality_lemma} we have  $\frac{1}{N_r}\bar{\mathbf{A}}_{ni,i}^{T}[k]\hat{\bar{\mathbf{A}}}_{jg,i}^{*}[k]=\mathbf{0}, \forall \{n,i\}\neq\{j,g\}$, and  $\frac{1}{N_r}\bar{\mathbf{A}}_{ni,i}^{T}[k]\hat{\bar{\mathbf{A}}}_{ni,i}^{*}[k]=\frac{1}{N_r}\text{diag}\{\mathbf{e}_{r,ni,i,k}(0) \hat{\mathbf{e}}_{r,ni,i,k}(0), \ldots, \mathbf{e}_{r,ni,i,k}(L_{ni,i}-1) \hat{\mathbf{e}}_{r,ni,i,k}(L_{ni,i}-1)\}$ for $\{n,i\}=\{j,g\}$. Hence, for massive MIMO systems, we can express the optimum eigen-beamformer under imperfect DoA estimation as
%\small
\begin{align}\label{V_ni_i_eig}
\tilde{\mathbf{V}}_{ni}^{eig}[k]=\frac{1}{N_r}\hat{\bar{\mathbf{A}}}_{ni,i}^{*}[k].
\end{align}
%
%\normalsize
Accordingly, achievable rate, in the presence of UL DoA estimation error, can be written as
%\small
\begin{align}\label{I_ni_k_hat}
\hat{\mathcal{I}}_{ni}[k]=\mathds{E}\left\{\sum\limits_{\ell=0}^{L_{ni,i}-1}\log\left(1+\hat{\gamma}_{ni,\ell}\left|\mathbf{e}_{r,ni,i,k}(\ell) \hat{\mathbf{e}}_{r,ni,i,k}(\ell)\right|^2\hat{p}_{ni,\ell}[k]\right)\right\},
\end{align}
%\normalsize
where $\hat{p}_{ni,\ell}[k]$ denotes the power to be allocated in the presence of DoA estimation error, $\hat{\gamma}_{ni,\ell}=\Lambda_{ni,i}|\alpha_{ni,i}(\ell)|^2/(N_r^2(\sigma^2+\zeta_{ni}))$, and the expectation is taken with respect to estimation error.
Using method of Lagrangian multiplier, the optimal expected power allocation on $\ell$-th information symbol for the $n$-th MS in the $i$-th cell on  the $k$-th subcarrier is given by
%\small
\begin{align}\label{E_p_ni_l_k}
\mathds{E}\{\hat{p}_{ni,\ell}[k]\}=\left[{\hat{\mu}}_{ni,\ell}[k]-\frac{1}{{\hat{\gamma}}_{ni,\ell}\mathds{E}\{|\mathbf{e}_{r,ni,i,k}^T(\ell)\mathbf{\hat{e}}_{r,ni,i,k}^*(\ell)|^2\}}\right]^\diamondsuit,
\end{align}
%\normalsize
where ${\hat{\mu}}_{ni,\ell}(k)$ is the corresponding Lagrange multiplier. Finally, \eqref{E_p_ni_l_k} can be simplified as \cite{Rubayet_Journal}:
%\small
\begin{align}\label{Proposed_power_allocation}
\mathds{E}\{\hat{p}_{ni,\ell}[k]\}=\left[{\hat{\mu}}_{ni,\ell}[k]-\frac{1}{ \hat{\gamma}_{ni,\ell} M_1^2M_2^2}{\left(1+\frac{M_1^2E\left[(\Delta v_{ni,i,\ell})^2\right]}{12}\right)}{\left(1+\frac{M_2^2E\left[(\Delta u_{ni,i,\ell})^2\right]}{12}\right)}\right]^\diamondsuit.
\end{align}
%\normalsize
It can be observed from that, in the absence of DoA estimation error, the optimal  power allocation algorithm in \eqref{Proposed_power_allocation} converges to  water filling solution in \eqref{traditional_water_filling}.
It is also to be noted here that both power allocations in \eqref{traditional_water_filling} and \eqref{Proposed_power_allocation} take into account the effects of inter-cell interference, unlike the single-user eigen-beamforming presented in \cite{Rubayet_Journal}.
{
\subsection{Precoding Complexity Analysis}
In this subsection, we briefly discuss the computational complexity of the proposed DoA-based precoding strategy as presented in Theorem \ref{Theorem_precoder}. Similarly to the Section  \ref{CompComplexity_DoA}, for computational complexity analysis, here we again assume that all the channels have $L$ resolvable paths. Now, for forming the effective channels, $\bar{\mathbf{D}}_{i,i}\bar{\mathbf{A}}_{i,i}^{T}[k]$, the number of FLOPS required is $E_a=2(JL-1)JLN_r$. Taking SVD of the effective channel requires $E_b=4J^2L^2N_r+22N_r^3$ FLOPS. Now, for constructing the final precoder, number of FLOPS required is $E_c=[2(N_r-1)N_rJL]+2(JL-1)N_rJL$. Hence, total number of FLOPS required for our DoA-based precoder can be written as $C_{\text{DoA}}=E_a+E_b+E_c$. Next, for comparison, we calculate the computational cost for conventional Block Diagonalization-based precoding. Let $\tilde{L}_j$ denote the rank of the matrix $[\mathbf{H}_{0i,i}^{{dl}^T},\ldots, \mathbf{H}_{(j-1)i,i}^{{dl}^T}, \mathbf{H}_{(j+1)i,i}^{{dl}^T}, \ldots, \mathbf{H}_{(J-1)i,i}^{{dl}^T}   ]^T $. For complexity analysis, we assume that $\tilde{L}_j=\tilde{L} ; \forall j$.
Hence, following the steps of conventional block diagonalization precoding, total number of FLOPS required for BD is $C_{\text{BD}} =J( [4 (J-1)^2N_t^2N_r+22N_r^3]+ [2(N_r-1)N_t(N_r-\tilde{L})] 
+[4N_t^2(N_r-\tilde{L})+22(N_r-\tilde{L})^3])$

}
\section{Performance Evaluation}
In this section, we evaluate the ESPRIT-based UL DoA estimation for multi-cell multi-user massive FD-MIMO OFDM networks through simulation.
For simulation evaluation, we consider seven hexagonal cells with MSs uniformly distributed in each cell.
Without loss of generality, we assume that number of co-scheduled MSs in each cell is $10$.
{An $M_1 \times M_2$  (antenna elements in elevation direction, and  antenna elements in the azimuth direction) rectangular antenna array is assumed at the BS, whereas the mobile device has a uniform linear array}
Different MSs are using non-orthogonal spreading sequences as UL pilots, and the same pool of sequences is reused in all seven cells.
%The PN-sequences are generated from a 5-stage linear feedback shift registers (LFSR), where the feedbacks from the registers are taken in such a way that it results in maximum length sequences (e.g. by taking the feedbacks from second and fifth stages).
Therefore, in the UL, the target BS is subject to intra-cell interference as well as interference from MSs in six other neighboring cells for the purpose of DoA estimation.
% As a performance metric, we use the root mean square error (RMSE) of the estimated elevation and azimuth angles, and RMSE of any angle (either azimuth or elevation) of path $\ell$ is defined as $\sqrt{\mathds{E}\{(\theta_{\ell}-\hat{\theta_{\ell}})^2\}}$, where $\theta_{\ell}$ and $\hat{\theta_{\ell}}$ denote true and estimated angles, respectively, and $\mathds{E}$ represents expectation over different channel realization.
Cell radius is set to be $1000$ meters.
{The system is assumed to operate at the mmWave band with 28 GHz carrier frequency.
}
$4$ dominant clusters are assumed for each UL channel from the MS to the BS, and each cluster contributes one resolvable path.
The antenna spacing for both the received and transmit antenna arrays is assumed to be $0.5 \lambda$.
The number of transmit antennas at each MS is set to be $8$.
In this paper, we invoke the far field assumption, and the wavefront impinging on the antenna array is assumed to be planer.
The transmission medium is assumed to be isotropic and linear.

\begin{figure}[!htb]
%	\centering
	\begin{minipage}{0.5\textwidth}
		\centering
		\includegraphics[width=1\linewidth]{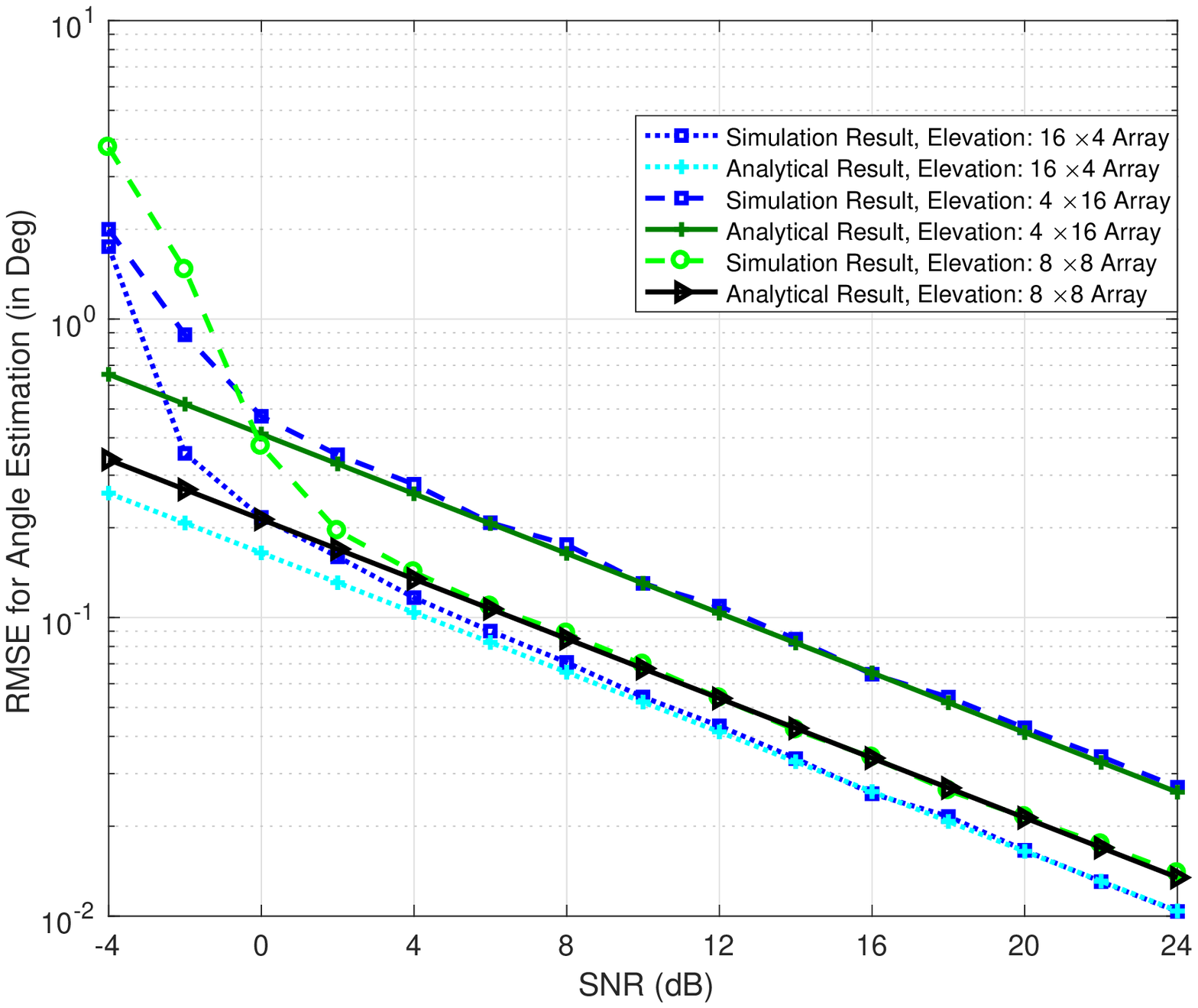}
		\caption{Elevation Angle Estimation for \\ 64 Antennas.}
		\label{64_elevation}
	\end{minipage}%
	\begin{minipage}{0.5\textwidth}
		\centering
		\includegraphics[width=1\linewidth]{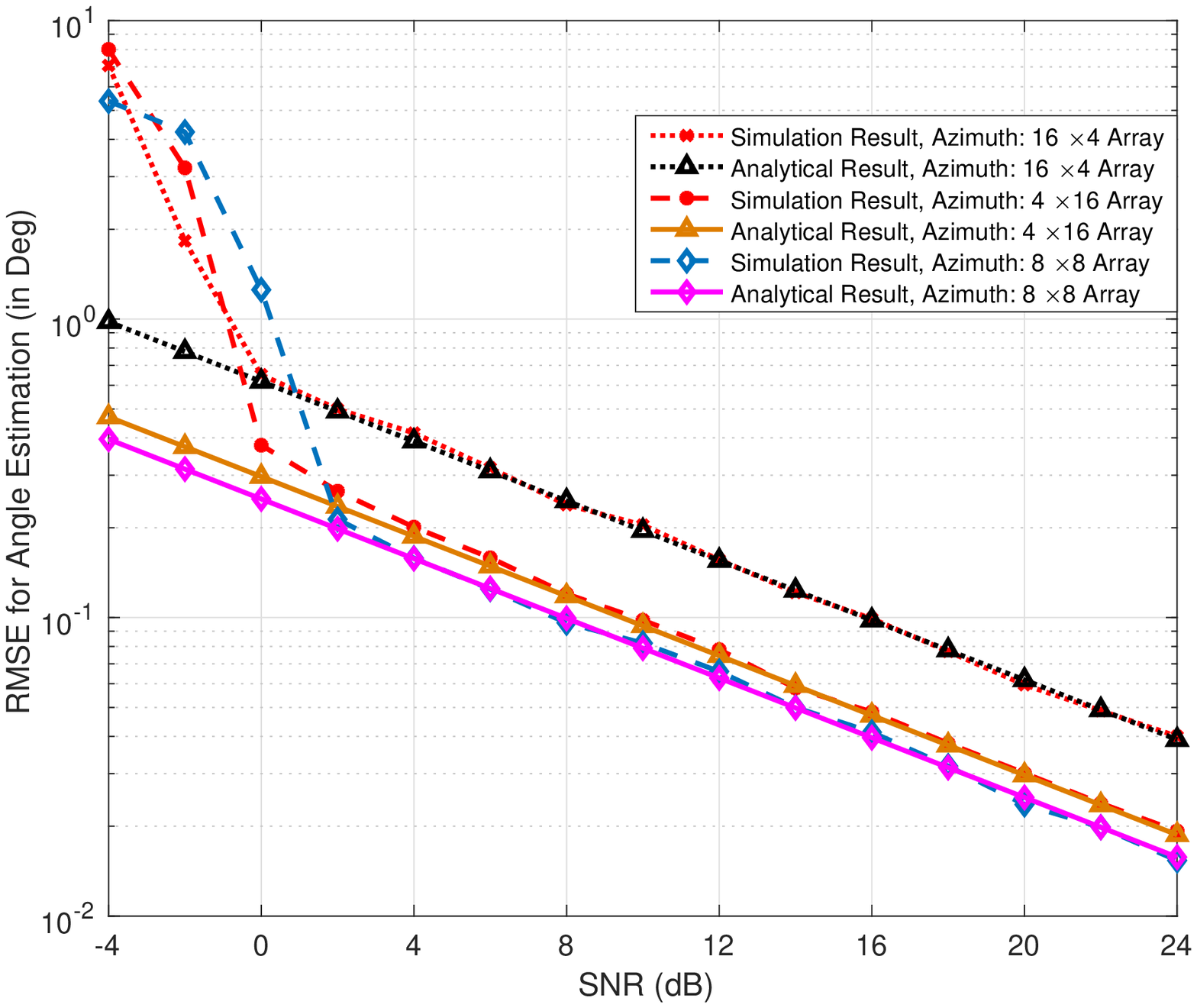}
		\caption{Azimuth Angle Estimation for 64 Antennas.}
		\label{64_azimuth}
	\end{minipage}
\end{figure}

%\begin{figure}
%\centering
%\includegraphics[width=0.5\linewidth]{Figures/64_elevation}
%\caption[width=\linewidth]{{Elevation Angle Estimation for 64 Antennas.}}
%\label{64_elevation}
%\end{figure}%

The estimation performance of elevation and azimuth angles for $8 \times 8$, $4 \times 16$,  and $16 \times 4$ antenna arrays are shown in Fig.~\ref{64_elevation} and Fig.~\ref{64_azimuth}, respectively, where the RMSE of the DoA estimation has been used as the performance metric, and the correlation coefficient of spreading sequences, $\rho_1$, is chosen to be $0.1$.
As the figure suggests, the analytical results of DoA estimation match well with that of empirical results asymptotically with SNR.
Furthermore, antenna array geometry has a significant impact on estimation performance.
Fig.~\ref{64_elevation} clearly suggests that the $16 \times 4$ antenna array performs better than $8 \times 8$ and $4 \times 16$ arrays in elevation angle estimation. 
However, $8 \times 8$ array configuration may outperform the $4 \times 16$ configuration in azimuth angle estimation as shown in Fig.~\ref{64_azimuth}. 
This is quite counter-intuitive since the $4 \times 16$ array has more elements in the azimuth domain.
The reason mainly comes from the fact that the azimuth DoA estimation is actually coupled with elevation DoA estimation. %which is evident from the Jacobian expressions relating the MSEs of spatial frequencies with elevation and azimuth angles \cite{Rubayet_Journal}.
For the $4 \times 16$ array, the performance of the elevation DoA estimation may be so bad that it affects the azimuth DoA  estimation performance.
{This dependence is manifested through Jacobians (see Remark \ref{jacobian11}), which, in fact, results from the underlying physics/ coordinate system of the 3D MIMO model.
}
On the other hand, elevation estimation is not dependent on azimuth estimation, and hence, $16 \times 4$ array geometry still outperforms $8 \times 8$ array in elevation angle estimation.
{These observations can provide important design intuitions for FD-MIMO networks that adopt subspace-based channel estimation methods.
}
%Moreover, at high SNR regime, increase of the total number of base station antennas also has a positive impact on the estimation of both elevation and azimuth angles, which is also justified by \cite{Rubayet_Journal}.
%We can also observe that antenna array geometry plays a significant role in determining the estimation performance. For the same total 64 antennas, $16 \times 4$ array performs better in elevation angle estimation whereas $8 \times 8$ outperforms $16 \times 4$ array in azimuth angle estimation.

%\begin{figure}
%\centering
%\includegraphics[width=0.5\linewidth]{Figures/64_azimuth}
%\caption[width=\linewidth]{{Azimuth Angle Estimation for 64 Antennas.}}
%\label{64_azimuth}
%\end{figure}
The elevation and azimuth angle estimation results for $16 \times 16$, $8 \times 32$, and $32 \times 8$ antenna arrays are shown in Fig.~\ref{256_elevation} and Fig.~\ref{256_azimuth}, respectively.
Comparing the results with those presented in Fig.~\ref{64_elevation} and Fig.~\ref{64_azimuth}, we can observe that as the total number of antennas increases, the DoA estimation accuracy accordingly increases, which is also evident from our analytical results.
%\begin{figure}
%	\centering
%	\includegraphics[width=0.5\linewidth]{Figures/256_elevation}
%	\caption[width=\linewidth]{{Elevation Angle Estimation for 256 Antennas.}}
%	\label{256_elevation}
%\end{figure}

\begin{figure}[!htb]
	%	\centering
	\begin{minipage}{0.5\textwidth}
		\centering
		\includegraphics[width=1\linewidth]{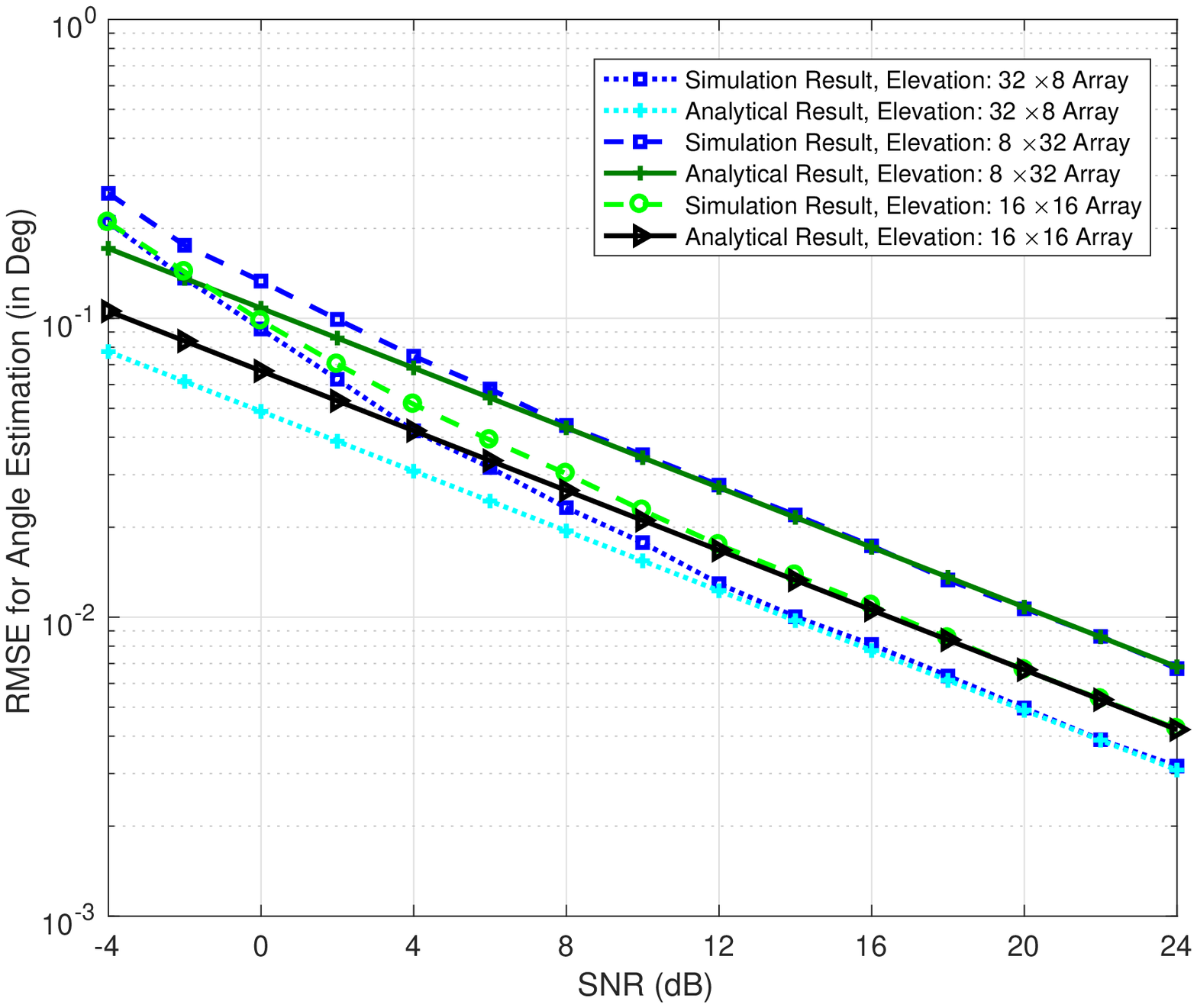}
		\caption{Elevation Angle Estimation for\\ 256 Antennas.}
		\label{256_elevation}
	\end{minipage}
	\begin{minipage}{0.5\textwidth}
		\centering
		\includegraphics[width=1\linewidth]{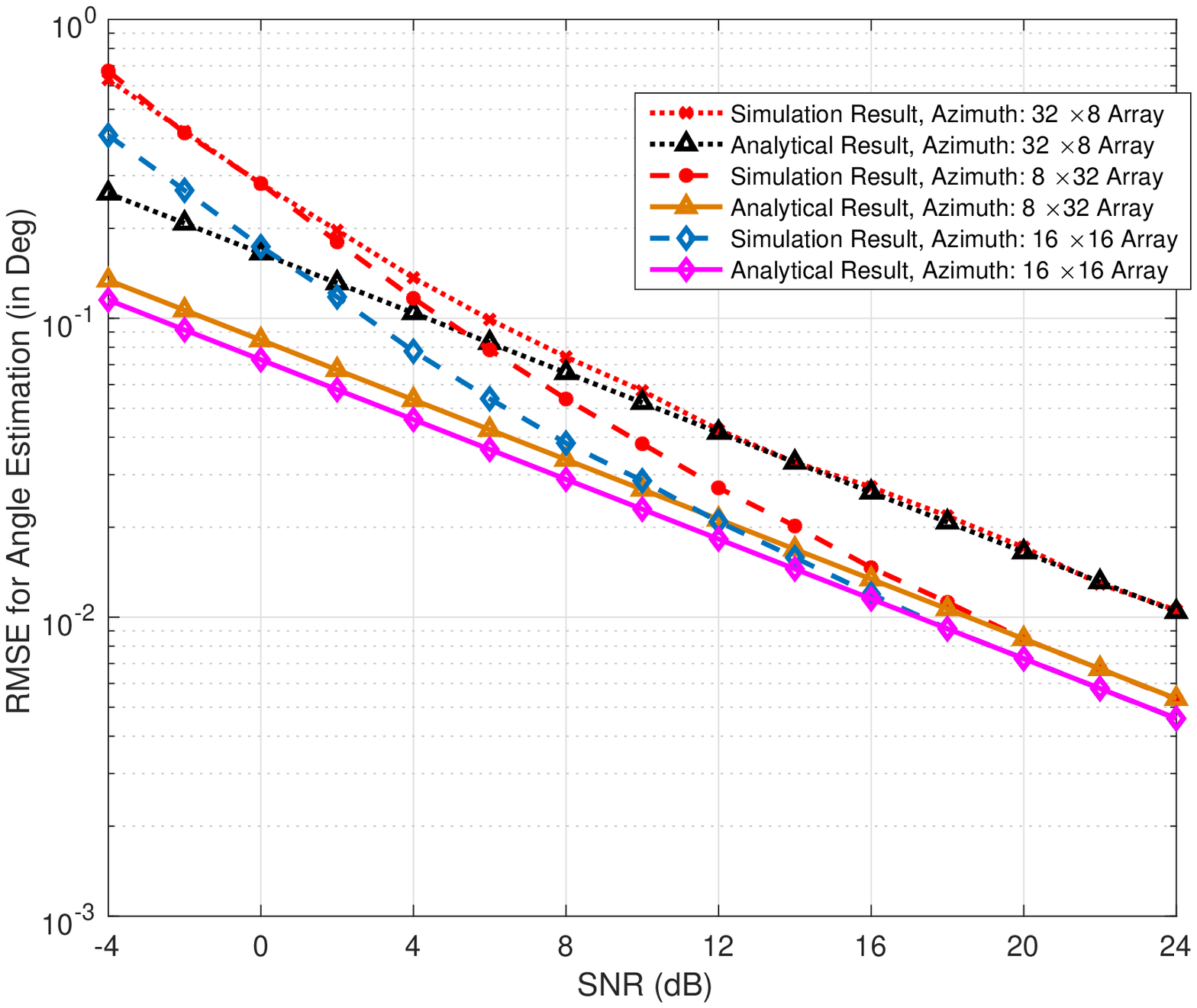}
		\caption{Azimuth Angle Estimation for \\ 256 Antennas.}
		\label{256_azimuth}
	\end{minipage}%
	
\end{figure}

\begin{figure}
	%\centering
%	\begin{minipage}{.5\textwidth}
%		\centering
%		\includegraphics[width=1.1\linewidth]{Figures/corr_8x8_v1_new}
%		%		\caption{figure}{A figure}
%		\caption[width=0.8\linewidth]{{Correlation between true and estimated channel.}}
%		\label{Corr_8x8}
%	\end{minipage}%
%	\begin{minipage}{.5\textwidth}
		\centering
		\includegraphics[width=0.5\textwidth]{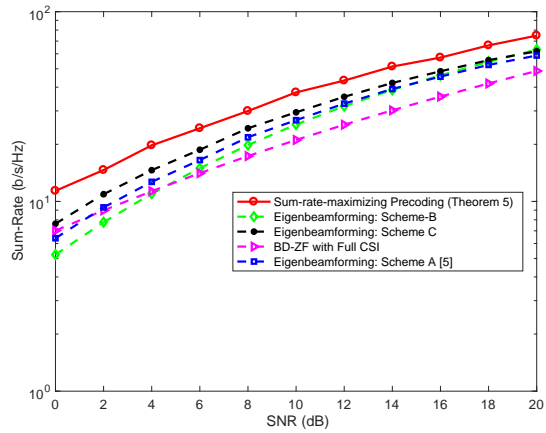}
		%		\caption{figure}{Another figure}
		\caption[width=\linewidth]{{Average Achievable Sum-Rate Comparison.}}
		\label{rate1}
%	\end{minipage}
\end{figure}
In Fig.~\ref{rate1}, the average achievable sum-rates for different precoding strategies are compared for multi-cell multi-user massive FD-MIMO networks.
Five schemes are compared: the introduced scheme presented in Theorem 5; the block-diagonalization based zero forcing (BD-ZF) precoding method~\cite{Joint_Power_Allocation_and_User, Massive_MIMO_for_Maximal_Spectral} assuming full CSI at the BS; and three eigen-beamforming schemes based on the large antenna system analysis.
To be specific, Scheme A is the single-user eigen-beamforming introduced in \cite{Rubayet_Journal}.
This scheme uses eigen-beamformer in~\eqref{V_ni_i_eig}, and applies the modified water-filling power allocation presented in~\cite{Rubayet_Journal} taking into account the DoA estimation error due to the noise.
However, Scheme A doesn't consider the effects of intra/inter-cell interference into power allocation.
In Scheme B, \eqref{V_ni_i_eig} is used as the beamformer and the traditional water-filling in~\eqref{traditional_water_filling} is used as power allocation assuming ideal DoA estimation.
Scheme C uses the same beamformer as Scheme A and B, however, it utilizes the power allocation in \eqref{Proposed_power_allocation} considering the DoA estimation error due to intra/inter-cell interference of the network.
Fig.~\ref{rate1} clearly suggests that the scheme introduced in Theorem 5 achieves best performance among all precoding strategies over the entire SNR regime of interests.
Even assuming full CSI at the BS, the BD-ZF scheme performs worst in the medium to high SNR regime.
This suggests that BD-ZF based precoding strategy may yield strictly suboptimal performance for massive FD-MIMO networks.
{ It is to be noted here that even though BD is using full channel state information, the performance gain of DoA-based method over BD method is coming from the fact that DoA-based method utilizes the structure of the underlying channel, whereas BD method does not take into account the underlying structure of the MIMO channel.}

For the three eigen-beamforming schemes based on the large antenna system analysis, we have the following observation:
Scheme C outperforms both Schemes A and B over the entire SNR regime since Scheme C considers the comprehensive characterization of the DoA estimation for power allocation as discussed in Remark 1.
Scheme A performs better than Scheme B at low SNRs indicating the importance of incorporating the DoA estimation error.
Since DoA estimation error decreases as SNR increases, both Scheme A and B approach Scheme C asymptotically.
This is because the power allocation in \eqref{Proposed_power_allocation} converges to water-filling solution in \eqref{traditional_water_filling} with increasing SNR.

{
\begin{figure}[!htb]
	%	\centering
	\begin{minipage}{0.5\textwidth}
		\centering
		\includegraphics[width=1\linewidth]{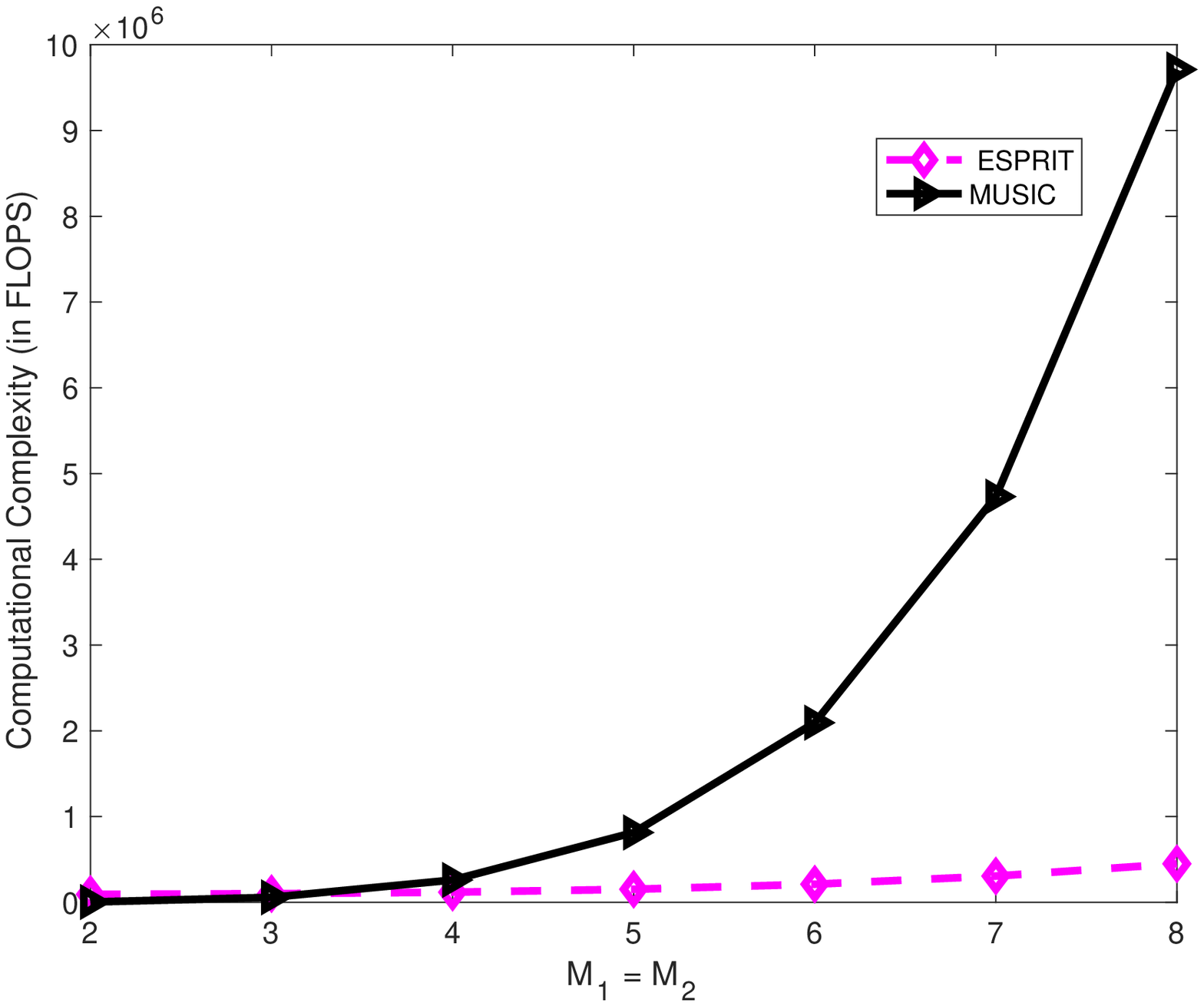}
		\caption{Computational Complexity \\Comparison for DoA Estimation Algorithms.}
		\label{FigComplexity_angle}
	\end{minipage}
	\begin{minipage}{0.5\textwidth}
		\centering
		\includegraphics[width=1\linewidth]{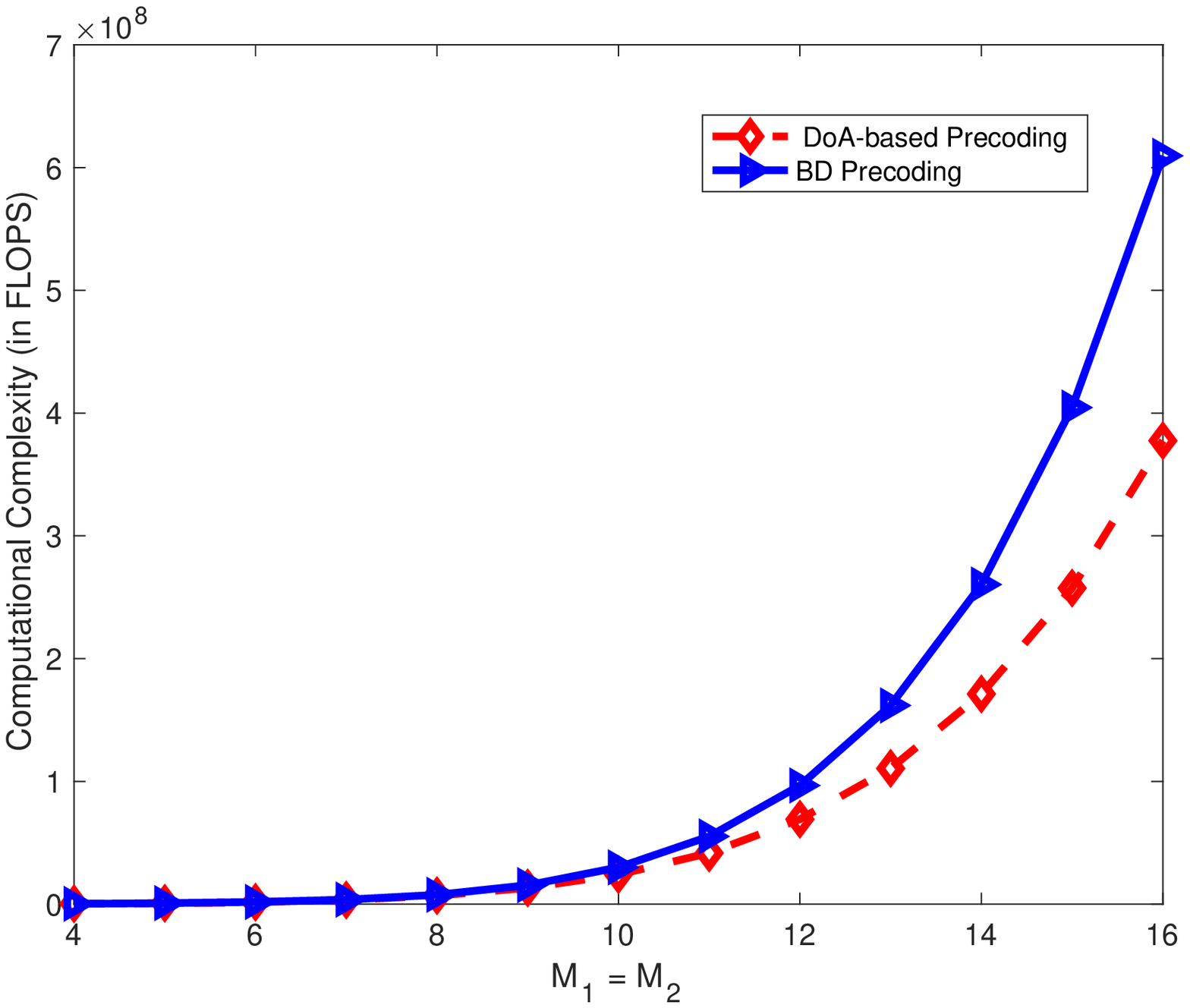}
		\caption{Computational Complexity \\Comparison for Precoding Methods.}
		\label{FigComplexity_Precoder}
	\end{minipage}%
	
\end{figure}
In Fig.~\ref{FigComplexity_angle} we compare computational complexity of our ESPRIT-based DoA estimation method with the widely used MUSIC algorithm, for a square BS antenna array (i.e., $M_1=M_2$). Number of transmit antennas is 8, and number of paths in the channel is 4. For MUSIC algorithm, number of grids for candidate DoA search is 360 which is a typical number. We can observe that as the number of antennas increases, the complexity of the MUSIC algorithm increases much faster compared to the complexity of ESPRIT algorithm. Hence, for massive MIMO system, MUSIC-based DoA estimation will incur significantly more computational burden than the ESPRIT method. 
In Fig.~\ref{FigComplexity_Precoder}, we compare the computational complexity of the our proposed DoA-based precoding scheme (Theorem \ref{Theorem_precoder}) is compared with the complexity of traditional Block Diagonalization (BD) precoding. We can observe that as for low and mid-size arrays, the complexity of both algorithms are similar. However, as the number of antennas increases the DoA-based precoder outperforms the BD method in term of the computational complexity.
}

%In Fig.~\ref{rate1} , we also compare the results with traditional block-diagonalization based zero  forcing (BD-ZF) precoding method \cite{Joint_Power_Allocation_and_User, Massive_MIMO_for_Maximal_Spectral} with full CSI.
%It can be seen that even though BD-ZF performs better than Scheme B eigen-beamforming at low SNR regime, at high SNR regime, BD-ZF performs worse than any eigen-beamforming methods. Finally, sum-rate maximization precoding in Theorem 5 outperforms all other methods over entire SNR regime.
\section{Conclusion}
Accurate DL CSI is critical for massive MIMO to realize the promised throughput gain.
In this paper, we introduced optimal DL MIMO precoding and power allocation strategies for multi-cell multi-user massive FD-MIMO networks based on UL DoA estimation at the BS.
The UL DoA estimation error for such a network has been analytically characterized and has been incorporated into the proposed MIMO precoding and power allocation strategy.
Simulation results suggested that the proposed strategy outperforms existing BD-ZF based MIMO precoding strategies which requires full CSI at the BS.
This work shed a light on system design for massive FD-MIMO communications which is critical for 5G and Beyond 5G cellular networks.

\appendices
\section{Proof of Theorem 1}\label{Proof_Theorem1}
Effect of pilot contamination on the MSE of DoA estimation is given by
%\small
\begin{align}\notag
\mathbb{E}\left\{\left(\triangle v_{ni,i,\ell} \right)^2\right\}_1 = &\frac{1}{2}\left({\mathbf{r}_{ni,i,\ell}^{(v)^{H}}}\cdot \mathbf{W}_{ni,i,mat}^* \cdot \mathbf{R}_{i,1}^{(fba)^{T}} \cdot \mathbf{W}_{ni,i,mat}^T \cdot {\mathbf{r}_{ni,i,\ell}^{(v)}} \right.\\\label{equ_mse1_pilot}
&\left.- \text{Re} \left\{{\mathbf{r}_{ni,i,\ell}^{(v)^{T}}}\cdot \mathbf{W}_{ni,i,mat} \cdot \mathbf{C}_{i,1}^{(fba)} \cdot \mathbf{W}_{ni,i,mat}^T \cdot {\mathbf{r}_{ni,i,\ell}^{(v)}}\right\}\right),
\end{align}
%\normalsize
Let us now denote
%\small
\begin{equation}
\bm{\beta}_{ni,i,\ell}=\mathbf{V}_{ni,i}^{sig} \mathbf{\Sigma}_{ni,i}^{{sig}^{-1}}\mathbf{q}_{\ell},
\end{equation}
%\normalsize
%\small
\begin{equation}
\bm{\alpha}_{v,{ni,i,\ell}}=\left(\mathbf{p}_{\ell}^T\left(\tilde{\mathbf{J}}_1^{(v)}\mathbf{U}_{ni,i}^{sig}\right)^{+}\left(\tilde{\mathbf{J}}_2^{(v)}/e^{jv_{ni,i,\ell}}-\tilde{\mathbf{J}}_1^{(v)}\right)\left(\mathbf{U}^{noise}_{ni,i} \mathbf{U}^{{noise}^H}_{ni,i}\right)\right)^T,
\end{equation}
%\normalsize
Using \eqref{r_ni_i_l} and \eqref{W_ni_i_mat}, we have $\mathbf{W}_{ni,i,mat}^T  {\mathbf{r}_{ni,i,\ell}^{(v)}}=\bm{\beta}_{ni,i,\ell}\otimes \bm{\alpha}_{v,{ni,i,\ell}}$. The MSE in \eqref{equ_mse1_pilot} becomes
%\small
\begin{align}\notag
%\begin{split}
\mathbb{E}\left\{\left(\triangle v_{ni,i,\ell} \right)^2\right\}_1  =& \frac{1}{2}\left(\left(\bm{\beta}_{ni,i,\ell}\otimes \bm{\alpha}_{v,{ni,i,\ell}}\right)^H \cdot \mathbf{R}_{i,1}^{(fba)T} \cdot \left(\bm{\beta}_{ni,i,\ell}\otimes \bm{\alpha}_{v,{ni,i,\ell}}\right)\right.\\
&\left.- \text{Re} \left\{\left(\bm{\beta}_{ni,i,\ell}\otimes \bm{\alpha}_{v,{ni,i,\ell}}\right)^T \cdot \mathbf{C}_{i,1}^{(fba)}\left(\bm{\beta}_{ni,i,\ell}\otimes \bm{\alpha}_{v,{ni,i,\ell}}\right)\right\}\right).
\label{equ_mse2}
\end{align}
%\normalsize
It can be easily verified that $\bm{\alpha}_{v,{ni,i,\ell}}$ can be written as
%\small
\begin{align}\notag
\bm{\alpha}_{v,{ni,i,\ell}}^T&=\mathbf{c}_\ell^T\left(\left(\mathbf{\tilde{J}}_{v,2}\mathbf{A}_{ni,i}\right)^{+}\mathbf{\tilde{J}}_{v,2}-
\left(\mathbf{\tilde{J}}_{v,1}\mathbf{A}_{ni,i}\right)^{+}\mathbf{\tilde{J}}_{v,1}\right),\\\notag
&=\frac{1}{(M_2-1)M_1}\left[
-1, -e^{-ju_{ni,i,\ell}}, \ldots,  -e^{-j(M_1-1)u_{ni,i,\ell}}, 0 , \ldots ,0,\right.\\
&\left. e^{-j(M_2-1)v_{ni,i,\ell}},  e^{-j((M_2-1)v_{ni,i,\ell}+u_{ni,i,\ell})}, \ldots, e^{-j((M_2-1)v_{ni,i,\ell}+(M_1-1)u_{ni,i,\ell})}
\right].
\end{align}
%\normalsize
Next, in order to obtain the expression for $\bm{\beta}_{ni,i,\ell}$, we need to perform the SVD of the perturbation-free  signal in~\eqref{H_nii_fba}:
%\small
\begin{equation*}
\begin{split}
&\sqrt{\Lambda_{ni,i}}\left[\mathbf{A}_{ni,i} \mathbf{D}_{ni,i}\mathbf{B}_{ni,i}^H(k)  \quad\bm{\Pi}_{N_r}\mathbf{A}_{ni,i}^{*} \mathbf{D}_{ni,i}^{*}\mathbf{B}_{ni,i}^T(k)\bm{\Pi}_{N_t}\right]\\
%&=\sqrt{\Lambda_{ni,i}}\left[\mathbf{A}_{ni,i} \mathbf{D}_{ni,i}\mathbf{B}_{ni,i}^H(k)  \quad\mathbf{A}_{ni,i} \bm{\Gamma}_{ni,i}\mathbf{D}_{ni,i}^{*}\mathbf{B}_{ni,i}^T(k)\bm{\Pi}_{N_t}\right]\\
&=\mathbf{A}_{ni,i}\text{diag}\left\{\mathbf{\bar{b}}_{ni,i}\right\}\left[\mathbf{\bar{B}}_{ni,i}^H(k) \quad\bm{\Gamma}_{ni,i}\mathbf{\bar{B}}_{ni,i}^T(k)\bm{\Pi}_{N_t}\right],
\end{split}
\end{equation*}
%\normalsize
where
\begin{equation*}
\begin{split}
\bm{\Gamma}_{ni,i}&=\text{diag}\left\{\left[e^{-j\left((M_1-1)u_{ni,i,0}+(M_2-1)v_{ni,i,0}\right)},\ldots,e^{-j\left((M_1-1)u_{ni,i,L_{ni,i}-1}+(M_2-1)v_{ni,i,L_{ni,i}-1}\right)}\right]\right\},\\
\mathbf{\bar{B}}_{ni,i}^H(k)&=\text{diag}\left\{e^{j{\phi}_{ni,i,0}^{'}},\ldots,e^{j{\phi}_{ni,i,L_{ni,i}-1}^{'}}\right\}\mathbf{B}_{ni,i}^H(k), \quad\text{and}\\
\mathbf{\bar{b}}_{ni,i}&=\left[b_{ni,i,0},\ldots,b_{ni,i,L_{ni,i}-1}\right],
\end{split}
\end{equation*}
where $b_{ni,i,\ell}$ and $\phi_{ni,i,\ell}^{'}$ are the amplitude and the phase of the channel gain $\alpha_{ni,i}(\ell)$, respectively.
Accordingly, based on \textbf{Lemma}~\ref{lemma:angle}, we can obtain
%\small
\begin{equation*}
\begin{split}
\mathbf{U}_{ni,i}^{sig}&=1/\sqrt{N_r}\mathbf{A}_{ni,i},\\ \mathbf{\Sigma}_{ni,i}^{sig}&=\sqrt{2N_rN_t}\sqrt{\Lambda_{ni,i}}\text{diag}\left\{\mathbf{\bar{b}}_{ni,i}\right\},\quad\text{and}\\ \mathbf{V}_{ni,i}^{{sig}^{H}}&= 1/\sqrt{2N_t}\left[\mathbf{\bar{B}}_{ni,i}^H(k) \quad \bm{\Gamma}_{ni,i}\mathbf{\bar{B}}_{ni,i}^T(k)\bm{\Pi}_{N_t}\right].
\end{split}
\end{equation*}
%\normalsize
In \cite{Performance_Analysis_for_DoA_Estimation}, the vector $\bm{\beta}_{ni,i,\ell} $ is given as $\bm{\beta}_{ni,i,\ell} = \mathbf{V}_{ni,i}^{sig}\mathbf{\Sigma}_{ni,i}^{{sig}^{-1}}\mathbf{U}_{ni,i}^{{sig}^{H}}\mathbf{A}_{ni,i}\mathbf{c}_\ell$. Now substituting here the expressions of $\mathbf{U}_{ni,i}^{sig}$, $\mathbf{\Sigma}_{ni,i}^{sig}$, and $\mathbf{V}_{ni,i}^{sig}$, we obtain:
%\small
\begin{equation}
\bm{\beta}_{ni,i,\ell} = \frac{1}{(b_{ni,i,\ell})\sqrt{\Lambda_{ni,i}}\sqrt{2N_t}}\mathbf{V}_{ni,i}^{sig}\mathbf{c}_{\ell}.
\label{equ_delaybeta3}
\end{equation}
%\normalsize
Hence, the expression $\left(\bm{\beta}_{ni,i,\ell}\otimes \bm{\alpha}_{v,{ni,i,\ell}}\right)$ in \eqref{equ_mse2} can be written as
%\small
\begin{equation}
\begin{split}
&\bm{\beta}_\ell\otimes\bm{\alpha}_{v,\ell}=\frac{1}{(b_{ni,i,\ell})2N_t\sqrt{\Lambda_{ni,i}}}
\left[\begin{array}{c}
e^{-j{\phi}_{ni,i,\ell}'}\mathbf{e}_{t,ni,i,k}(\ell) \\
e^{j{\phi}_{ni,i,\ell}'}e^{j\left((M_1-1)u_{ni,i,\ell}+(M_2-1)v_{ni,i,\ell}\right)}\bm{\Pi}_{N_t}\mathbf{e}_{t,ni,i,k}^*(\ell)
\end{array}
\right]\otimes\bm{\alpha}_{v,\ell}\\
&= \frac{1}{(b_{ni,i,\ell})2N_t\sqrt{\Lambda_{ni,i}}}
\left[\begin{array}{c}
e^{-j{\phi}_{ni,i,\ell}'}\mathbf{e}_{t,ni,i,k}(\ell)\otimes\bm{\alpha}_{v,ni,i,\ell} \\
e^{j{\phi}_{ni,i,\ell}'}e^{j\left((M_1-1)u_{ni,i,\ell}+(M_2-1)v_{ni,i,\ell}\right)}\bm{\Pi}_{N_t}\mathbf{e}_{t,ni,i,k}^*(\ell)\otimes\bm{\alpha}_{v,ni,i,\ell}
\end{array}
\right],
\label{equ_albe1}
\end{split}
\end{equation}
%\normalsize
Now, using equation \eqref{equ_covariance_1} and \eqref{equ_albe1}, the first term in \eqref{equ_mse2} can be written as
%\small
\begin{align}\notag
&\left(\bm{\beta}_{ni,i,\ell}\otimes \bm{\alpha}_{v,{ni,i,\ell}}\right)^H \cdot \mathbf{R}_{i,1}^{(fba)T} \cdot \left(\bm{\beta}_{ni,i,\ell}\otimes \bm{\alpha}_{v,{ni,i,\ell}}\right)\\\notag
&=\frac{1}{b_{ni,i,\ell}^24N_t^2\Lambda_{ni,i}}\bigg(
\left(\mathbf{e}_{t,ni,i,k}^H(\ell)\otimes\bm{\alpha}_{v,ni,i,\ell}^H\right)\mathbf{R}_{i,1}^T\left(\mathbf{e}_{t,ni,i,k}(\ell)\otimes\bm{\alpha}_{v,ni,i,\ell}\right)\\\label{First_part_MSE1_pilot}
&+\left(\mathbf{e}_{t,ni,i,k}^T(\ell)\mathbf{\Pi}_{N_t}\otimes\bm{\alpha}_{v,ni,i,\ell}^H\right)\mathbf{\Pi}_{N_rN_t}\mathbf{R}_{i,1}^H\bm{\Pi}_{N_rN_t}\left(\mathbf{\Pi}_{N_t}\mathbf{e}_{t,ni,i,k}^*(\ell)\otimes\bm{\alpha}_{v,ni,i,\ell}\right)\bigg)
\end{align}
%\normalsize
Now, using \eqref{R_i_1_k}, and after some simplification, we have
%\small
\begin{align}\notag
&\left(\mathbf{e}_{t,ni,i,k}^H(\ell)\otimes\bm{\alpha}_{v,ni,i,\ell}^H\right)\mathbf{R}_{i,1}^T\left(\mathbf{e}_{t, ni,i,k}(\ell)\otimes\bm{\alpha}_{v,ni,i,\ell}\right)\\\label{FirstPart_1_Pilot}
&=\frac{1}{(M_2-1)^2M_1^2}\sum\limits_{\substack{g=0\\g\neq i}}^{G-1}\left(\sqrt{\Lambda_{ng,i}}\right)^2X_{ng,i}Y_{ng,i}\sum\limits_{m=1}^{L}|\alpha_{ng,i}(m)|^2
\end{align}
%\normalsize
where  $X_{ng,i}$ and $Y_{ng,i}$ are given by
%\small
\begin{align}
&X_{ng,i}=\mathds{E}_{\psi}\left|\left(1+e^{-j(\omega_{ni,i,\ell}-\omega_{ng,i,m})}+\ldots+e^{-j(N_t-1)(\omega_{ni,i,\ell}-\omega_{ng,i,m})}\right)\right|^2,\\
&Y_{ng,i}=\mathds{E}_{\theta,\phi}\left|\left(1+e^{j(u_{ni,i,\ell}-u_{ng,i,m})}+\ldots+e^{j(M_1-1)(u_{ni,i,\ell}-u_{ng,i,m})}\right)\left(e^{jv_{ni,i,\ell}}e^{-jv_{ng,i,m}}-1\right)\right|^2,
\end{align}
%\normalsize
for $m=0,\ldots L_{ng,i}-1$, and $\mathds{E}_{\psi}$ and $\mathds{E}_{\theta,\phi}$ denote, respectively, expectations with respect to DoD and DoAs. Now, similarly to \eqref{FirstPart_1_Pilot}, we also have
%\small
\begin{align}\notag
&\left(\mathbf{e}_{t,ni,i,k}^T(\ell)\mathbf{\Pi}_{N_t}\otimes\bm{\alpha}_{v,ni,i,\ell}^H\right)\mathbf{\Pi}_{N_rN_t}\mathbf{R}_{i,1}^H\bm{\Pi}_{N_rN_t}\left(\mathbf{\Pi}_{N_t}\mathbf{e}_{t,ni,i,k}^*(\ell)\otimes\bm{\alpha}_{v,ni,i,\ell}\right)\\\label{FirstPart_2_Pilot}
&=\frac{1}{(M_2-1)^2M_1^2}\sum\limits_{\substack{g=0\\g\neq i}}^{G-1}\left(\sqrt{\Lambda_{ng,i}}\right)^2X_{ng,i}Y_{ng,i}^{'}\sum\limits_{m=1}^{L}|\alpha_{ng,i}(m)|^2,
\end{align}
%\normalsize
where
%\small
\begin{align}\notag
Y_{ng,i}^{'}=&\mathds{E}_{\theta,\phi}\left|\left(e^{j(M_1-1)u_{ng,i,m}}+e^{ju_{ni,i,\ell}}e^{j(M_1-2)u_{ng,i,m}}+\ldots+e^{j(M_1-1)u_{ni,i,\ell}}\right)\times\right.\\
&\qquad\left.\left(e^{j(M_2-1)v_{ni,i,\ell}}-e^{j(M_2-1)v_{ng,i,m}}\right)\right|^2,
\end{align}
%\normalsize
for $m=0,\ldots L_{ng,i}-1$. Now, using \eqref{FirstPart_1_Pilot} and \eqref{FirstPart_2_Pilot}, we can write \eqref{First_part_MSE1_pilot} as
%\small
\begin{align}\notag
& \left(\bm{\beta}_{ni,i,\ell}\otimes \bm{\alpha}_{v,{ni,i,\ell}}\right)^H \cdot \mathbf{R}_{i,1}^{(fba)T} \cdot \left(\bm{\beta}_{ni,i,\ell}\otimes \bm{\alpha}_{v,{ni,i,\ell}}\right)\\\label{First_part_MSE1_pilot_simplified1}
&=\frac{1}{b_{ni,i,\ell}^24N_t^2\Lambda_{ni,i}}\frac{1}{(M_2-1)^2M_1^2}\sum\limits_{\substack{g=0\\g\neq i}}^{G-1}\left(\sqrt{\Lambda_{ng,i}}\right)^2X_{ng,i}\sum\limits_{m=0}^{L_{ng,i}-1}|\alpha_{ng,i}(m)|^2\left(Y_{ng,i}+Y_{ng,i}^{'}\right)
\end{align}
%\normalsize
Similarly, we can also have
%\small
\begin{align}\notag
&\left(\bm{\beta}_{ni,i,\ell}\otimes \bm{\alpha}_{v,{ni,i,\ell}}\right)^T \cdot \mathbf{C}_{i,1}^{(fba)}\left(\bm{\beta}_{ni,i,\ell}\otimes \bm{\alpha}_{v,{ni,i,\ell}}\right)\\\label{second_part_MSE1_pilot_simplified1}
&=\frac{1}{b_{ni,i,\ell}^22N_t^2\Lambda_{ni,i}}\frac{1}{(M_2-1)^2M_1^2}e^{j\Phi}
\sum\limits_{\substack{g=0\\g\neq i}}^{G-1}\left(\sqrt{\Lambda_{ng,i}}\right)^2X_{ng,i}\tilde{Y}_{ng,i}\sum\limits_{m=0}^{L_{ng,i}-1}|\alpha_{ng,i}(m)|^2,
\end{align}
%\normalsize
where $\Phi=\left((M_1-1)u_{ni,i,\ell}+(M_2-1)v_{ni,i,\ell}\right)$, and $\tilde{Y}_{ng,i}$ is given in \eqref{Y_ng_i_tilde}.
%\small
%\begin{align}\notag
%&\tilde{Y}_{ng,i}=\mathds{E}_{\theta,\phi}\bigg[\left(e^{-j(M_1-1)u_{ng,i,m}}+e^{-ju_{ni,i,\ell}}e^{-j(M_1-2)u_{ng,i,m}}+\ldots+e^{-j(M_1-2)u_{ni,i,\ell}}e^{-ju_{ng,i,m}}+e^{-j(M_1-1)u_{ni,i,\ell}}\right)\\
%&\left(e^{-j(M_2-1)v_{ni,i,\ell}}-e^{-j(M_2-1)v_{ng,i,m}}\right)\left(1+e^{j(u_{ng,i,m}-u_{ni,i,\ell})}+\ldots+e^{(M_1-1)(u_{ng,i,m}-u_{ni,i,\ell})}\right)\left(e^{j(M_2-1)(v_{ng,i,m}-v_{ni,i,\ell})}-1\right)\bigg]
%\end{align}
%\normalsize
Finally, plug the expressions from \eqref{First_part_MSE1_pilot_simplified1} and \eqref{second_part_MSE1_pilot_simplified1} into \eqref{equ_mse2}, and the proof is finished.
%\small
%\begin{align}
%&\mathds{E}\{(\Delta v_{\ell})^2\}_{1}=\frac{1}{b_{ni,i,\ell}^28N_t^2\Lambda_{ni,i}(M_2-1)^2M_1^2}
%\sum\limits_{\substack{g=0\\g\neq i}}^{G-1}\left(\sqrt{\Lambda_{ng,i}}\right)^2X_{ng,i}\sum\limits_{m=1}^{L}|\alpha_{ng,i}(m)|^2\left(Y_{ng,i}+Y_{ng,i}^{'}-2\Re\{e^{j\Phi}\tilde{Y}_{ng,i}\}\right)
%\end{align}
%\normalsize

\section{Proof of Theorem 2}\label{Proof_Theorem2}
Similarly to \eqref{equ_mse2}, MSE due to intra-cell interference can be written as
%\small
\begin{align}\notag
%\begin{split}
\mathbb{E}\left\{\left(\triangle v_{ni,i,\ell} \right)^2\right\}_2  =& \frac{1}{2}\left(\left(\bm{\beta}_{ni,i,\ell}\otimes \bm{\alpha}_{v,{ni,i,\ell}}\right)^H \cdot \mathbf{R}_{i,2}^{(fba)T} \cdot \left(\bm{\beta}_{ni,i,\ell}\otimes \bm{\alpha}_{v,{ni,i,\ell}}\right)\right.\\
&\left.- \text{Re} \left\{\left(\bm{\beta}_{ni,i,\ell}\otimes \bm{\alpha}_{v,{ni,i,\ell}}\right)^T \cdot \mathbf{C}_{i,2}^{(fba)}\left(\bm{\beta}_{ni,i,\ell}\otimes \bm{\alpha}_{v,{ni,i,\ell}}\right)\right\}\right).
\label{MSE_intraCell_transformed}
\end{align}
%\normalsize
Using \eqref{equ_covariance_1} for $m=2$, the first term in \eqref{MSE_intraCell_transformed} can be expressed as
%\small
\begin{align}\notag
&\left(\bm{\beta}_{ni,i,\ell}\otimes \bm{\alpha}_{v,{ni,i,\ell}}\right)^T \cdot \mathbf{C}_{i,2}^{(fba)}\left(\bm{\beta}_{ni,i,\ell}\otimes \bm{\alpha}_{v,{ni,i,\ell}}\right)\\\notag
&=\frac{1}{b_{ni,i,\ell}^24N_t^2\Lambda_{ni,i}}e^{j\Phi}\left[\left(\mathbf{e}_{t,ni,i,k}^H(\ell)\bm{\Pi}_{N_t}\otimes\bm{\alpha}_{v,ni,i,\ell}^T\right)\bm{\Pi}_{N_rN_t}\mathbf{R}_{i,2}^*\left(\mathbf{e}_{t,ni,i,k}(\ell)\otimes\bm{\alpha}_{v,ni,i,\ell}\right)\right.\\\label{SecondPart_Intra_modified}
&\left.+\left(\mathbf{e}_{t,ni,i,k}^T(\ell)\otimes\bm{\alpha}_{v,ni,i,\ell}^T\right)\mathbf{R}_{i,2}\bm{\Pi}_{N_rN_t}\left(\bm{\Pi}_{N_t}\mathbf{e}_{t,ni,i,k}^*(\ell)\otimes\bm{\alpha}_{v,ni,i,\ell}\right)\right]
\end{align}
%\normalsize
Now, using \eqref{R_i_2_k}, and after some simplifications, we can write the first term in \eqref{SecondPart_Intra_modified} as
%\small
\begin{align}\notag
&\left(\mathbf{e}_{t,ni,i,k}^H(\ell)\bm{\Pi}_{N_t}\otimes\bm{\alpha}_{v,ni,i,\ell}^T\right)\bm{\Pi}_{N_rN_t}\mathbf{R}_{i,2}^*\left(\mathbf{e}_{t,ni,i,k}(\ell)\otimes\bm{\alpha}_{v,ni,i,\ell}\right)\\\notag
=&\rho_1^2\mathds{E}_{\alpha,\theta,\phi,\psi}\left[\sum\limits_{\substack{j=0\\j\neq n}}^{J-1}\left(\sqrt{\Lambda_{ji,i}}\right)^2\left(\mathbf{e}_{t,ni,i,k}^H(\ell)\mathbf{1}_{N_t}\mathbf{B}_{ji,i}(k)\otimes\bm{\alpha}_{v,ni,i,\ell}^T\bm{\Pi}_{N_r}\mathbf{A}_{ji,i}^*\right)\text{vec}\left\{\mathbf{D}_{ji,i}\right\}^*\right.\\\label{SecondPart_Intra_modified_firstpart}
&\left.\times\text{vec}\left\{\mathbf{D}_{ji,i}\right\}^T\left(\mathbf{B}_{ji,i}^H(k)\mathbf{1}_{N_t}\mathbf{e}_{t,ni,i,k}(\ell)\otimes\mathbf{A}_{ji,i}^T\bm{\alpha}_{v,ni,i,\ell}\right)\vphantom{\sum\limits_{\substack{j=0\\j\neq n}}^{J-1}}\right]
\end{align}
%\normalsize
It can be shown that
%\small
\begin{align}\notag
&\mathds{E}_{\alpha}\left[\left(\mathbf{e}_{t,ni,i,k}^H(\ell)\mathbf{1}_{N_t}\mathbf{B}_{ji,i}(k)\otimes\bm{\alpha}_{v,ni,i,\ell}^T\bm{\Pi}_{N_r}\mathbf{A}_{ji,i}^*\right)\text{vec}\left\{\mathbf{D}_{ji,i}\right\}^*\text{vec}\left\{\mathbf{D}_{ji,i}\right\}^T\times\right.\\\notag
&\quad\left.\left(\mathbf{B}_{ji,i}^H(k)\mathbf{1}_{N_t}\mathbf{e}_{t,ni,i,k}(\ell)\otimes\mathbf{A}_{ji,i}^T\bm{\alpha}_{v,ni,i,\ell}\right)\right]\\\label{withRespectToAlpha_1}
=&\left|X_{ni,i,\ell}^{''}\right|^2\sum\limits_{m=0}^{L_{ji,i}-1}\left|\alpha_{ji,i}(m)\right|^2\left|X_{ji,i}(m)\right|^2\bm{\alpha}_{v,ni,i,\ell}^T\bm{\Pi}_{N_r}\mathbf{e}_{ji,i}^*(m)\mathbf{e}_{ji,i}^T(m)\bm{\alpha}_{v,ni,i,\ell},
\end{align}
%\normalsize
where $X_{ni,i,\ell}^{''}=\sum\limits_{r=0}^{N_t-1}e^{jr\omega_{ni,i,\ell}}$, and $X_{ji,i}(m)=\sum\limits_{r=0}^{N_t-1}e^{jr(\omega_{ni,i,\ell}-\omega_{ji,i,m}) }$.
%Now, let $\mathds{E}_{\psi}\left[\left|X_{ji,i}(m)\right|^2\right]=\tilde{X}_{ji,i}$.
After some tedious but straight forward calculations, we have
%\small
\begin{align}\label{withRespectToThetaPhi_1}
\mathds{E}_{\theta,\phi}\left[\bm{\alpha}_{v,ni,i,\ell}^T\bm{\Pi}_{N_r}\mathbf{e}_{ji,i}^*(m)\mathbf{e}_{ji,i}^T(m)\bm{\alpha}_{v,ni,i,\ell}\right]=\frac{1}{\left(M_2-1\right)^2M_1^2}\tilde{Y}_{ji,i}.
\end{align}
%\normalsize
%where
%\small
%\begin{align}\notag
%\tilde{Y}_{ji,i}&=\mathds{E}_{\theta,\phi}\left[\left(e^{-j(M_1-1)u_{ji,i}}+e^{-ju_{ni,i,\ell}}e^{-j(M_1-2)u_{ji,i}}+\ldots+e^{-j(M_1-2)u_{ni,i,\ell}}e^{-ju_{ji,i}}+e^{-j(M_1-1)u_{ni,i,\ell}}\right)\right.\\
%&\left.\times\left(e^{-j(M_2-1)v_{ni,i,\ell}}-e^{-j(M_2-1)v_{ji,i}}\right)\left(1+\ldots+e^{j(M_1-1)(u_{ji,i}-u_{ni,i,\ell})}\right)\left(e^{j(M_2-1)(v_{ji,i}-v_{ni,i,\ell})}-1\right)\right]
%\end{align}
%\normalsize
Now, using \eqref{withRespectToAlpha_1} and \eqref{withRespectToThetaPhi_1}, we can simplify \eqref{SecondPart_Intra_modified_firstpart} as follows:
%\small
\begin{align}\notag
&\left(\mathbf{e}_{t,ni,i,k}^H(\ell)\bm{\Pi}_{N_t}\otimes\bm{\alpha}_{v,ni,i,\ell}^T\right)\bm{\Pi}_{N_rN_t}\mathbf{R}_{i,2}^*\left(\mathbf{e}_{t,ni,i,k}(\ell)\otimes\bm{\alpha}_{v,ni,i,\ell}\right)\\\label{SecondPart_Intra_modified_firstpart_simplified}
&=\rho_1^2\frac{1}{(M_2-1)^2M_1^2}{|X_{ni,i,\ell}^{''}|}^2\left[\sum\limits_{\substack{j=0\\j\neq n}}^{J-1}\left(\sqrt{\Lambda_{ji,i}}\right)^2X_{ji,i}\tilde{Y}_{ji,i}\left(\sum\limits_{m=0}^{L_{ji,i}-1}|\alpha_{ji,i}(m)|^2\right)\right].
\end{align}
%\normalsize
Similarly, we can simplify the second term in \eqref{SecondPart_Intra_modified} as follows:
%\small
\begin{align}\notag
&\left(\mathbf{e}_{t,ni,i,k}^H(\ell)\bm{\Pi}_{N_t}\otimes\bm{\alpha}_{v,ni,i,\ell}^T\right)\bm{\Pi}_{N_rN_t}\mathbf{R}_{i,2}^*\left(\mathbf{e}_{t,ni,i,k}(\ell)\otimes\bm{\alpha}_{v,ni,i,\ell}\right)\\\label{SecondPart_Intra_modified_secondpart_simplified}
&=\rho_1^2\frac{1}{(M_2-1)^2M_1^2}{|X_{ni,i,\ell}^{''}|}^2\left[\sum\limits_{\substack{j=0\\j\neq n}}^{J-1}\left(\sqrt{\Lambda_{ji,i}}\right)^2X_{ji,i}\tilde{Y}_{ji,i}\left(\sum\limits_{m=0}^{L_{ji,i}-1}|\alpha_{ji,i}(m)|^2\right)\right].
\end{align}
%\normalsize
%where
%\small
%\begin{align}\notag
%\tilde{Y}_{ji,i}^{'}&=\mathds{E}_{\theta,\phi}\left[\left(e^{-j(M_1-1)u_{ni,i,\ell}}+e^{-ju_{ji,i}}e^{-j(M_1-2)u_{ni,i,\ell}}+\ldots+e^{-j(M_1-2)u_{ji,i}}e^{-ju_{ni,i,\ell}}+e^{-j(M_1-1)u_{ji,i}}\right)\right.\\
%&\left.\left(e^{-j(M_2-1)v_{ni,i,\ell}}-e^{-j(M_2-1)v_{ji,i}}\right)\left(1+\ldots+e^{-j(M_1-1)(u_{ni,i,\ell}-u_{ji,i})}\right)\left(e^{-j(M_2-1)(v_{ni,i,\ell}-v_{ji,i})}-1\right)\right].
%\end{align}
%\normalsize
Now, plugging the expressions from \eqref{SecondPart_Intra_modified_firstpart_simplified} and \eqref{SecondPart_Intra_modified_secondpart_simplified} into \eqref{SecondPart_Intra_modified}, we obtain
%\small
%\begin{align}\notag
%&\left(\bm{\beta}_{ni,i,\ell}\otimes \bm{\alpha}_{v,{ni,i,\ell}}\right)^T \cdot \mathbf{C}_{i,2}^{(fba)}\left(\bm{\beta}_{ni,i,\ell}\otimes \bm{\alpha}_{v,{ni,i,\ell}}\right)\\\label{SecondPart_Intra_modified_simplified}
%&=\frac{1}{b_{ni,i,\ell}^24N_t^2\Lambda_{ni,i}}\rho_1^2\frac{1}{(M_2-1)^2M_1^2}{|X_{ni,i,\ell}^{''}|}^2e^{j\Phi}\left[\sum\limits_{\substack{j=0\\j\neq n}}^{J-1}\left(\sqrt{\Lambda_{ji,i}}\right)^2\left(\sum\limits_{m=0}^{L_{ji,i}-1}|\alpha_{ji,i}(m)|^2\right)\tilde{X}_{ji,i}\left(\tilde{Y}_{ji,i}+\tilde{Y}_{ji,i}^{'}\right)\right]
%\end{align}
%\normalsize
%\small
\begin{align}\notag
&\left(\bm{\beta}_{ni,i,\ell}\otimes \bm{\alpha}_{v,{ni,i,\ell}}\right)^T \cdot \mathbf{C}_{i,2}^{(fba)}\left(\bm{\beta}_{ni,i,\ell}\otimes \bm{\alpha}_{v,{ni,i,\ell}}\right)\\\notag
=&\frac{1}{b_{ni,i,\ell}^24N_t^2\Lambda_{ni,i}}\rho_1^2\frac{1}{(M_2-1)^2M_1^2}{|X_{ni,i,\ell}^{''}|}^2e^{j\Phi}\times\\\label{SecondPart_Intra_modified_simplified}
&\quad\left[\sum\limits_{\substack{j=0\\j\neq n}}^{J-1}\left(\sqrt{\Lambda_{ji,i}}\right)^2\left(\sum\limits_{m=0}^{L_{ji,i}-1}|\alpha_{ji,i}(m)|^2\right)X_{ji,i}\left(2\tilde{Y}_{ji,i}\right)\right]
\end{align}
%\normalsize
Following similar procedure, we can also obtain
%\small
\begin{align}\notag
&\left(\bm{\beta}_{ni,i,\ell}\otimes \bm{\alpha}_{v,{ni,i,\ell}}\right)^H \cdot \mathbf{R}_{i,2}^{(fba)T} \cdot \left(\bm{\beta}_{ni,i,\ell}\otimes \bm{\alpha}_{v,{ni,i,\ell}}\right)\\\notag
=&\frac{1}{b_{ni,i,\ell}^24N_t^2\Lambda_{ni,i}}\rho_1^2\frac{1}{(M_2-1)^2M_1^2}{|X_{ni,i,\ell}^{''}|}^2\times\\\label{FirstPart_Intra_modified_simplified}
&\quad\left[\sum\limits_{\substack{j=0\\j\neq n}}^{J-1}\left(\sqrt{\Lambda_{ji,i}}\right)^2\left(\sum\limits_{m=0}^{L_{ji,i}-1}|\alpha_{ji,i}(m)|^2\right)X_{ji,i}\left(Y_{ji,i}+Y_{ji,i}^{'}\right)\right],
\end{align}
%\normalsize
%where $\bar{Y}_{ji,i}$ and $\bar{Y}_{ji,i}^{'}$ are given by
%\small
%\begin{align}
%\bar{Y}_{ji,i}&=\mathds{E}_{\theta,\phi}\left[\left|\left(e^{j(M_1-1)u_{ji,i}}+e^{ju_{\ell}}e^{j(M_1-2)u_{ji,i}}+\ldots+e^{j(M_1-1)u_{ni,i,\ell}}\right)\left(e^{j(M_2-1)v_{ni,i,\ell}}-e^{-j(M_2-1)v_{ji,i}}\right)\right|^2\right],
%\end{align}
%\normalsize
%\small
%\begin{align}
%\bar{Y}_{ji,i}^{'}&=\mathds{E}_{\theta,\phi}\left[\left|\left(1+e^{j(u_{ni,i,\ell}-u_{ji,i})}+\ldots+e^{j(M_1-1)(u_{ni,i,\ell}-u_{ji,i})}\right)\left(e^{j(M_2-1)(v_{ni,i,\ell}-v_{ji,i})}-1\right)\right|^2\right].
%\end{align}
%\normalsize
Now, plugging the expressions from  \eqref{SecondPart_Intra_modified_simplified} and \eqref{FirstPart_Intra_modified_simplified} into \eqref{MSE_intraCell_transformed}, we obtain the desired result.
%\small
%\begin{align}\notag
%%\begin{split}
%&\mathbb{E}\left\{\left(\triangle v_{ni,i,\ell} \right)^2\right\}_2 \\ &=\frac{\rho_1^2{|X_{ni,i,\ell}^{''}|}^2}{b_{ni,i,\ell}^28N_t^2\Lambda_{ni,i}(M_2-1)^2M_1^2}\sum\limits_{\substack{j=0\\j\neq n}}^{J-1}\left(\sqrt{\Lambda_{ji,i}}\right)^2\left(\sum\limits_{m=1}^{L}|\alpha_{ji,i}(m)|^2\right)\tilde{X}_{ji,i}\left[\bar{Y}_{ji,i}+\bar{Y}_{ji,i}^{'}-\Re\left\{e^{j\Phi}\left(\tilde{Y}_{ji,i}+\tilde{Y}_{ji,i}^{'}\right)\right\}\right],
%\end{align}
%\normalsize

\section{Proof of Theorem 5}\label{Proof_Theorem5}
The problem in \eqref{MSE_problem3} is a convex quadratic optimization problem, and can be solved using  Lagrangian method.
The Lagrange function for \eqref{MSE_problem3} can be written as
%\small
\begin{align}\notag
L\left(\mathbf{V}_{i}[k],\mu_{ik}\right)&=\text{Tr}\{\mathbf{R}_{i}^H\mathbf{H}_{i,i}[k]\mathbf{V}_{i}[k]\mathbf{V}_{i}^H[k]\mathbf{R}_{i}-\mathbf{R}_{i}^H\mathbf{H}_{i,i}[k]\mathbf{V}_{i}[k]-\mathbf{V}_{i}^H[k]\mathbf{R}_i+\mathbf{I}\}\\
& \ \ \ +\mu_{ik}\left(\text{Tr}\{\mathbf{V}_{i}[k]\mathbf{V}_{i}^H[k]\}-P_t\right),
\end{align}
%\normalsize
where $\mu_{ik}$ is the corresponding Lagrange multiplier. Now, taking the derivative of the Lagrange function w.r.t. $\mathbf{V}_{i}[k]$ and setting the derivative equal to zero, we can obtain the desired result.
\bibliographystyle{IEEEtran}
\bibliography{IEEEabrv,MIMO_OFDM}
\end{document}